\documentclass[11pt]{article}
\usepackage{amsmath, amssymb, amsthm, amsfonts, graphicx,color} 
\usepackage{cite}
\usepackage[normalem]{ulem}
\usepackage{geometry} 
\usepackage{graphicx}
\usepackage{amscd}

\def\({\left(}
\def\){\right)}
\def\1{\mathbf{1}}

\def\diam{\mathrm{diam}\ }
\def\dist{\text{dist}\ }

\def\dt0{{{\frac{d}{dt}}_{|t=0}}}

\def\ep{\varepsilon}

\def\l|{\left|}

\def\r|{\right|}

\numberwithin{equation}{section}

\newcommand{\eps}{\varepsilon} \newcommand{\TT}{{\mathbb T^2_\ell}}

\newtheorem{theorem}{Theorem}
\newtheorem{coro}{Corollary}[section] 
\newtheorem{lem}[coro]{Lemma}

\newtheorem{definition}[coro]{Definition}
\newtheorem{remark}[coro]{Remark}

\begin{document}

\title{\bf The $\mathbf \Gamma$-limit of the two-dimensional
  Ohta-Kawasaki energy. I. Droplet density.}

\author{Dorian Goldman, Cyrill B. Muratov and Sylvia Serfaty}

\maketitle

\begin{abstract}
  This is the first in a series of two papers in which we derive a
  $\Gamma$-expansion for a two-dimensional non-local Ginzburg-Landau
  energy with Coulomb repulsion, also known as the Ohta-Kawasaki model
  in connection with diblock copolymer systems.  In that model, two
  phases appear, which interact via a nonlocal Coulomb type energy. We
  focus on the regime where one of the phases has very small volume
  fraction, thus creating small ``droplets'' of the minority phase in
  a ``sea" of the majority phase.  In this paper we show that an
  appropriate setting for $\Gamma$-convergence in the considered
  parameter regime is via weak convergence of the suitably normalized
  charge density in the sense of measures. We prove that, after a
  suitable rescaling, the Ohta-Kawasaki energy functional
  $\Gamma$-converges to a quadratic energy functional of the limit
  charge density generated by the screened Coulomb kernel.  A
  consequence of our results is that minimizers (or almost minimizers)
  of the energy have droplets which are almost all asymptotically
  round, have the same radius and are uniformly distributed in the
  domain.  The proof relies mainly on the analysis of the sharp
  interface version of the energy, with the connection to the original
  diffuse interface model obtained via matching upper and lower bounds
  for the energy. We thus also obtain an asymptotic characterization
  of the energy minimizers in the diffuse interface model.
\end{abstract}

\section{Introduction}\label{sec:intro}

In the studies of energy-driven pattern formation, one often
encounters variational problems with competing terms operating on
different spatial scales
\cite{seul95,hubert,huebener,strukov,vedmedenko,muthukumar97,
  lundqvist}.  Despite the fundamental importance of these problems to
a multitude of physical systems, their detailed mathematical studies
are fairly recent (see
e.g. \cite{sandier,kohn07iciam,friesecke06,choksi99,choksi01,
  choksi08,choksi04,choksi98}). To a great extent this fact is related
to the emerging multiscale structure of the energy minimizing patterns
and the associated difficulty of their description
\cite{muller93,choksi99,desimone00,choksi08,lebris05}. In particular,
the popular approach of $\Gamma$-convergence \cite{braides} is
rendered difficult due to the emergence of more than two
well-separated spatial scales in suitable asymptotic limits (see
e.g. \cite{choksi98,choksi99,choksi04,choksi08,desimone00,lebris05,
  sandier12,muller93}).  % Therefore, the natural setting for
% $\Gamma$-convergence in the context of pattern forming systems with
% competing short-range and long-range interactions is not
% obvious. Alternatively, a more refined idea of $\Gamma$-expansion
% \cite{braides08} and of the method of ``lower bounds for 2-scale
% energies'' via $\Gamma$-convergence, introduced in \cite{sandier12},
% recently provided some successes. How to apply this approach to the
% multiscale setting of pattern forming systems is also not clear. In
% this series of works, we tackle these issues for the Ohta-Kawasaki
% functional, which can be considered mathematically a

These issues can be readily seen in the case of the Ohta-Kawasaki
model, a canonical mathematical model in the studies of energy-driven
pattern forming systems.  This model, originally proposed in
\cite{ohta86} to describe different morphologies observed in diblock
copolymer melts (see e.g. \cite{bates99}) is defined (up to a choice of scales)
  by the energy
functional \begin{align}
  \label{EE}
  \mathcal{E}[u] = \int_\Omega \( \frac{\eps^2}{2} |\nabla u|^2 +
  W(u)\) dx + \frac{1}{2} \int_\Omega \int_\Omega (u(x)- \bar u)
  G_0(x, y) (u(y)- \bar u) \, dx \, dy,
\end{align}
where $\Omega$ is the domain occupied by the material, $u: \Omega \to
\mathbb R$ is the scalar order parameter, $W(u)$ is a symmetric
double-well potential with minima at $u = \pm 1$, such as the usual
Ginzburg-Landau potential $W(u) = \tfrac14 (1 - u^2)^2$, $\eps > 0$ is
a parameter characterizing interfacial thickness, $\bar u \in (-1, 1)$
is the background charge density, and $G_0$ is the Neumann Green's
function of the Laplacian, i.e., $G_0$ solves
\begin{align}
  \label{G0}
  -\Delta G_0(x, y) = \delta(x - y) - {1 \over |\Omega|}, \qquad
  \int_\Omega G_0(x, y) \, dx = 0,
\end{align}
where $\Delta$ is the Laplacian in $x$ and $\delta(x)$ is the Dirac
delta-function, with Neumann boundary conditions. Note that $u$ is
also assumed to satisfy the ``charge neutrality'' condition
\begin{align}
  \label{neutr}
  {1 \over |\Omega|} \int_\Omega u \, dx = \bar u.
\end{align}
Let us point out that in addition to a number of polymer systems
\cite{degennes79,stillinger83,nyrkova94}, this model is also
applicable to many other physical systems due to the fundamental
nature of the Coulombic non-local term in \eqref{EE}
\cite{lundqvist,emery93,chen93,glotzer95,nagaev95,m:pre02}.  Because
of this Coulomb interaction, we also like to think of $u$ as a density
of ``charge''.  

The Ohta-Kawasaki functional admits the following ``sharp-interface''
version:
\begin{align}
  \label{E}
  E[u] = \frac{\eps}{2} \int_\Omega |\nabla u| \, dx +
  \frac12\int_\Omega \int_\Omega (u(x) - \bar u) G(x, y) (u(y) - \bar
  u) \, dx \, dy,
\end{align}
where now $u : \Omega \to \{-1, +1\}$ and $G(x, y)$ is the {\em
  screened} Green's function of the Laplacian, i.e., it solves the
Neumann problem for the equation (distinguish from \eqref{G0})
\begin{align}
  \label{G}
  -\Delta G + \kappa^2 G = \delta(x - y), 
\end{align}
where $\kappa := 1/\sqrt{W''(1)} > 0$. Note also that in contrast to
the diffuse interface energy in \eqref{EE}, for the sharp interface
energy in \eqref{E} the charge neutrality constraint in \eqref{neutr}
is no longer imposed. This is due to the fact that in a minimizer of
the diffuse interface energy,  the charge of the minority phase is
expected to partially redistribute into the majority phase to ensure
screening of the induced non-local field (see a more detailed
discussion in the following section).  

The two terms in the energy \eqref{E} are competing: the second term
favors $u$ to be constant and equal to its average $\bar u$, but since
$u$ is valued in $\{+1, -1\}$ this means in effect that it is
advantageous for $u$ to oscillate rapidly between the two phases $u =
+1$ and $u = -1$; the first term penalizes the perimeter of the
interface between the two phases, and thus opposes too much spreading
and oscillation. The competition between these two selects a length
scale, which is a function of $\ep$.  In the diffuse interface version
\eqref{EE}, the sharp transitions between $\{u= +1\}$ and $\{u=-1\}$
are replaced by smooth transitions at the scale $\ep > 0$ as soon as
$\ep \ll 1$.

\begin{figure}
  \centering
  \includegraphics[width=13cm]{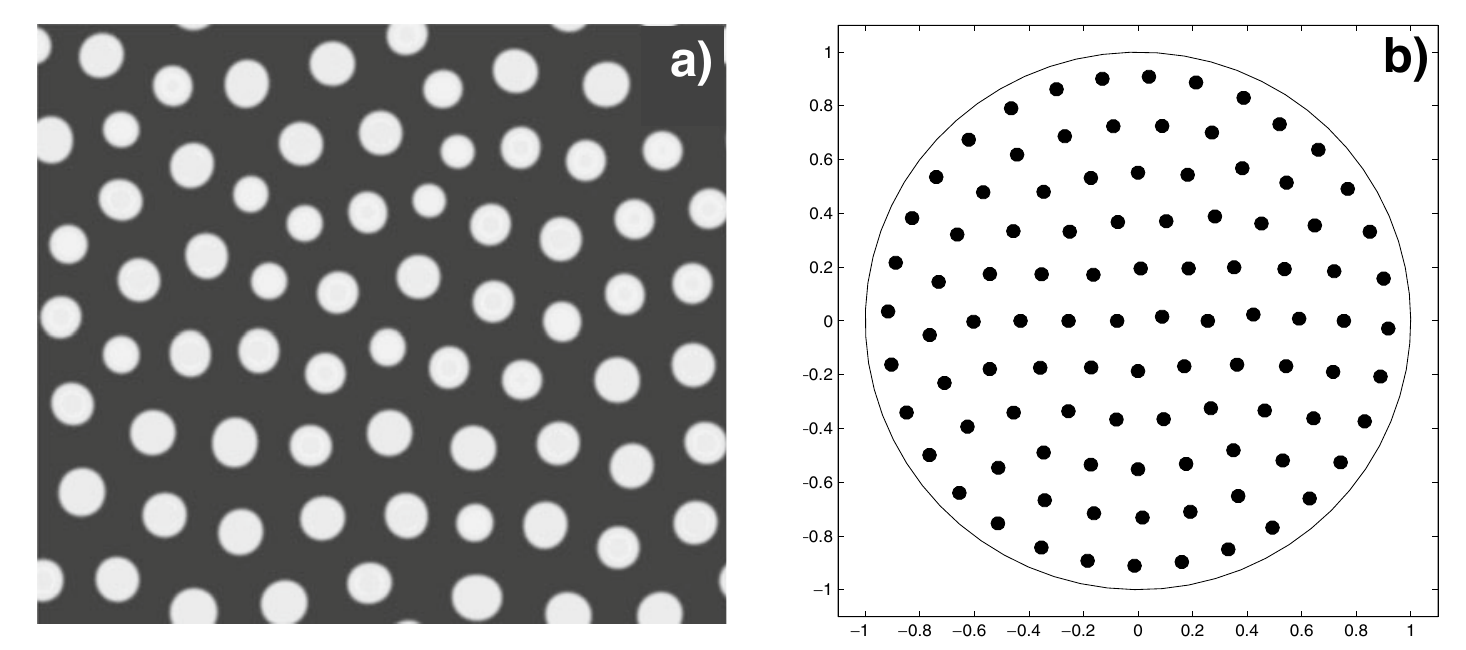}
  \caption{Two-dimensional multi-droplet patterns in systems with
    Coulombic repulsion: a local minimizer of the Ohta-Kawasaki energy
    on a rectangle with periodic boundary conditions; a local
    minimizer of the sum of two-point Coulombic potentials on a disk
    with Neumann boundary conditions. Taken from
    \cite{m:pre02,ren07rmp}. \label{f:drops} }
\end{figure}

In one space dimension and in the particular case $\bar u = 0$
(symmetric phases) the behavior of the energy can be understood from
the work of M\"uller \cite{muller93}: the minimizer $u$ is periodic
and alternates between $u = +1$ and $u = -1$ at scale $\ep^{1/3}$ (for
other one-dimensional results, see also
\cite{ren00trusk,ren00,yip06}). In higher dimensions the patterns of
minimizers are much more complex and are not well understood.  The
behavior depends on the volume fraction between the phases, i.e. on
the constant $\bar u$ chosen, and also on the dimension.  When
$\bar{u}<0$, we call $u=-1$ the majority phase and $u=+1$ the minority
phase, and conversely when $\bar{u}>0$. In two dimensions, numerical
simulations lead to expecting round ``droplets'' of the minority phase
surrounded by a ``sea'' of the majority phase (see Fig. \ref{f:drops})
for sufficient asymmetries between the majority and the minority
phases (i.e., for $\bar u$ sufficiently far away from zero)
\cite{ohta86,m:phd,m:pre02,ren07rmp}. The situation is less clear for
$\bar u$ close to zero, although it is commonly believed that in this
case the minimizers are one-dimensional stripe patterns
\cite{ohta86,m:phd,m:pre02,choksi11siads}.

In all cases, minimizers are intuitively expected to be periodic.
However, at the moment this seems to be very difficult to prove.  The
only general result in that direction to date is that of Alberti,
Choksi and Otto \cite{alberti09}, which proves that the energy of
minimizers of the sharp interface energy from \eqref{E} with no
screening (with $\kappa = 0$ and the neutrality condition from
\eqref{neutr}) is uniformly distributed in the limit where the size of
the domain $\Omega$ goes to infinity (see also
\cite{choksi01,spadaro09}). Their results, however, do not provide any
further information about the structure of the energy-minimizing
patterns.  Note in passing that the question of proving any
periodicity of minimizers for multi-dimensional energies is unsolved
even for systems of point particles forming simple crystals (see
e.g. \cite{lebris05,sandier12}), with a notable exception of certain
two-dimensional particle systems with short-range interactions which
somehow reduce to packing problems
\cite{theil06,wagner83,radin81}. Naturally, the situation can be
expected to be more complicated for pattern forming systems in which
the constitutive elements are ``soft'' objects, such as, e.g.,
droplets of the minority phase in the matrix of the majority phase in
the Ohta-Kawasaki model.

Here we are going to focus on the two-dimensional case and the
situation where one phase is in strong majority with respect to the
other, which is imposed by taking $\bar{u}$ very close to $-1$ as $\ep
\to 0$.  Thus we can expect a distribution of small droplets of $u=+1$
surrounded by a sea of $u = -1$.  In this regime, Choksi and Peletier
analyzed the asymptotic properties of a suitably rescaled version of
the sharp interface energy \eqref{E} with no screening in
\cite{choksi10}, as well as \eqref{EE} in \cite{choksi11}.  They work
in the setting of a fixed domain $\Omega$, and in a regime where the
number of droplets remains finite as $\ep \to 0$. They showed that the
energy minimizing patterns concentrate to a finite number of point
masses, whose magnitudes and locations are determined via a
$\Gamma$-expansion of the energy \cite{braides08}.  Here, in contrast,
we work in a regime where the number of droplets is divergent as $\ep
\to 0$.  We note that $\Gamma$-convergence of \eqref{EE} to the
functional \eqref{E} with no screening and for fixed volume fractions
was established by Ren and Wei in \cite{ren00}, who also analyzed
local minimizers of the sharp interface energy in the strong asymmetry
regime in two space dimensions \cite{ren07rmp}.

All these works are in the finite domain $\Omega$ setting, while we
are generally interested in the {\em large volume} (macroscopic)
limit, i.e., the regime when the number of droplets tends to infinity.
A rather detailed study of the behavior of the minimizers for the
Ohta-Kawasaki energy in macroscopically large domains was recently
performed in \cite{m:cmp10}, still in the regime of $\bar{u}$ close to
$-1$.  There the two-dimensional Ohta-Kawasaki energy was considered
in the case when $\Omega$ is a unit square with periodic boundary
conditions. The interesting regime corresponds to the parameters $\eps
\ll 1$ and $1 + \bar u = O(\eps^{2/3} |\ln \eps|^{1/3}) \ll 1$. It is
shown in \cite{m:cmp10} that under these assumptions on the parameters
and some technical assumptions on $W$, \eqref{E} gives the correct
asymptotic limit of the minimal energy in \eqref{EE}. Moreover, it is
shown that when $\bar \delta := \eps^{-2/3} |\ln \eps|^{-1/3} (1 +
\bar u)$ becomes greater than a certain critical constant
$\bar{\delta}_c$, the minimizers of $E$ in \eqref{E} consist of
$O(|\ln \eps|)$ simply connected, nearly round droplets of radius
$\simeq 3^{1/3} \eps^{1/3} |\ln \eps|^{-1/3}$, and uniformly
distributed throughout the domain \cite{m:cmp10}. Thus, the following
hierarchy of length scales is established in the considered regime:
\begin{align}
  \label{hier}
  \eps \ll \eps^{1/3} |\ln \eps|^{-1/3} \ll |\ln \eps|^{-1/2} \ll 1, 
\end{align}
where the scales above correspond to the width of the interface, the
radius of the droplets, the average distance between the droplets, and
the screening length, respectively. The multiscale nature of the
energy minimizing pattern is readily apparent from \eqref{hier}.

The analysis of \cite{m:cmp10} makes heavy use of the minimality
condition for \eqref{E} and, in particular, the Euler-Lagrange
equation associated with the energy. One is thus naturally led to
asking whether the qualitative properties of the minimizers
established in \cite{m:cmp10} (roundness of the droplets, identical
radii, uniform distribution) carry over to, e.g., almost minimizers of
$E$, for which no Euler-Lagrange equation is available. More broadly,
it is natural to ask how robust the properties of the energy
minimizing patterns are with respect to various perturbations of the
energy, for example, how the picture presented above is affected when
the charge density $\bar u$ is spatially modulated. A natural way to
approach these questions is via $\Gamma$-convergence. However, for a
multiscale problem such as the one we are considering the proper
setting for studying $\Gamma$-limits of the functionals in \eqref{EE}
or \eqref{E} is presently lacking. The purpose of this paper is to
formulate such a setting and extract the leading order term in the
$\Gamma$-expansion of the energy in \eqref{EE}. In our forthcoming
paper \cite{gms11b}, we obtain the next order term in the
$\Gamma$-expansion, using the method of ``lower bounds for 2-scale
energies'' via $\Gamma$-convergence introduced in \cite{sandier12}.

The main question for setting up the $\Gamma$-limit in the present
context is to choose a suitable metric for $\Gamma$-convergence. This
metric turns out to be similar to the one used for the analysis of
vortices in the two-dimensional magnetic Ginzburg-Landau model from
the theory of superconductivity \cite{sandier}. In fact, the problem
under consideration and its mathematical treatment (here as well as in
\cite{gms11b}) share several important features with the latter
\cite{sandier}. In the theory of superconductivity the role of
droplets is played by the Ginzburg-Landau vortices, which in the
appropriate limits also become uniformly distributed throughout the
domain \cite{sandier00}.  We note, however, that the approach
developed in \cite{sandier00,sandier} cannot be carried over directly
to the problem under consideration, since the vortices are more rigid
than their droplet counterparts: the topological degrees of the
vortices are quantized and can only take integer values, while the
droplet volumes are not. Thus we also have to consider the possibility
of many very small droplets. Developing a control on the droplet
volumes from above and below is one of the key ingredient of the
proofs presented below, and relies on the control of their perimeter
via the energy.

For simplicity, as in \cite{m:cmp10} we consider the energy defined on
a flat torus (a square with periodic boundary conditions).  The metric
we consider is the weak convergence of measures for a suitably
rescaled sequence of characteristic functions associated with droplets
(see the next section for precise definitions and statements of
theorems). Then, up to a rescaling, we show that both the energy
$\mathcal E$ from \eqref{EE} and $E$ from \eqref{E} $\Gamma$-converge
to a quadratic functional in terms of the limit measure, with the
quadratic term generated by the screened Coulomb kernel from \eqref{G}
and the linear term depending explicitly on $\bar\delta$ and
$\kappa$. To be more precise, we will see that in the regime we study,
there are two contributions to the energy which operate at leading
order: one contribution is linear in the density of the droplets and
corresponds to the ``self-interaction energy'' of each droplet coming
from both the perimeter term and self-interaction part of the double
integral in \eqref{E}, and the other is a quadratic term corresponding
to the interaction between the droplets, i.e. the rest of the
contribution of the double-integral term in \eqref{E}. This setting,
where both terms are of the same order of magnitude is very similar to
the regime of \cite{sandier00} and \cite[Chap. 7]{sandier} in the
context of the magnetic Ginzburg-Landau energy.

We note that the obtained limit variational problem is strictly convex
and its unique minimizer is a measure with constant density across the
domain $\Omega$. In particular, this implies equidistribution of mass
and energy for the minimizers of the diffuse interface energy
$\mathcal E$ in \eqref{EE} in the considered regime. In our companion
paper \cite{gms11b}, we further address the mutual arrangement of the
droplets in the energy minimizing patterns, using the formalism
developed recently for Ginzburg-Landau vortices \cite{sandier12}.  We
also obtain a characterization of the droplet shapes for almost
minimizers of the sharp interface energy $E$, which, in turn, allows
us to make the same conclusions about minimizers of the diffuse
interface energy $\mathcal E$ for $\eps \ll 1$, which is a new
result. The reason we can characterize the droplets at the diffuse
interface level is because the difference between the zero superlevel
set of the minimizers at the diffuse interface level and the jump set
of almost minimizers at the sharp interface level occurs essentially
on the length scale $\eps$ (interfacial thickness), which is much
smaller than the characteristic length scale $\eps^{1/3} |\ln
\eps|^{-1/3}$ of the droplets.

Let us mention other closely related systems from the studies of
ferromagnetism and superconductivity, where the role of droplets is
played by the slender needle-like domains of opposite magnetization in
a three-dimensional ferromagnetic slab at the onset of magnetization
reversal \cite{km:jns11}, or superconducting tunnels in a slab of
type-I superconducting material near the critical field
\cite{choksi04,choksi08}. It may be possible to obtain similar
$\Gamma$-convergence results with respect to convergence of measures
in the plane for those problems. At the same time, we point out that
extending our results to higher dimensions meets with serious
difficulties, since in the suitable limit the droplets in
higher-dimensional problems are expected to solve a non-local
isoperimetric problem whose solution is not well characterized at
present \cite{km:cpam13}.

Our paper is organized as follows. In
Sec. \ref{sec:statement-results}, we introduce the considered scaling
regime and state our main results; in Sec. \ref{sec:main} we prove the
$\Gamma$-convergence result in the sharp interface setting; in
Sec. \ref{sec:proof-theorem-almost} we prove the results on the
characterization of almost minimizers of sharp interface energy; and
in Sec. \ref{sec:proof-theor-refm} we treat the $\Gamma$-limit for the
case of the diffuse interface energy.

\paragraph{Some notations.} We use the notation $(u^\eps) \in \mathcal
A$ to denote sequences of functions $u^\eps \in \mathcal A$ as $\eps =
\eps_n \to 0$, where $\mathcal A$ is an admissible class.  For a
measurable set $E$, we use $|E|$ to denote its Lebesgue measure and
$|\partial E|$ to denote its perimeter (in the sense of De Giorgi). We
also use the notation $\mu \in \mathcal M^+ (\Omega)$ to denote a
non-negative Radon measure $\mu$ on the domain $\Omega$. With a slight
abuse, we will often speak of $\mu$ as the ``density'' on
$\Omega$. The symbols $H^1(\Omega)$, $BV(\Omega)$, $C(\Omega)$ and
$H^{-1}(\Omega)$ denote the usual Sobolev space, space of functions of
bounded variation, space of continuous functions, and the dual of
$H^1(\Omega)$, respectively.

\section{Statement of results}
\label{sec:statement-results}

Throughout the rest of the paper the parameters $\kappa > 0$,
$\bar\delta > 0$ and $\ell > 0$ are assumed to be fixed, and the
domain $\Omega$ is assumed to be a flat two-dimensional torus of side
length $\ell$, i.e., $\Omega = \mathbb T^2_\ell = [0, \ell)^2$, with
periodic boundary conditions. For every $\eps > 0$ we define
\begin{align}
  \label{ubeps}
  \bar u^\eps := -1 + \eps^{2/3} |\ln \eps|^{1/3} \bar\delta.
\end{align}
Under this scaling assumption the sharp interface version of the
Ohta-Kawasaki energy (cf. \eqref{E}) can be written as
\begin{equation}
  \label{E2}
  E^\eps[u] = \frac{\varepsilon}{2} \int_\TT  |\nabla u| \, dx +
  \frac{1}{2}  \int_\TT   (u-\bar u^\eps) (-\Delta +
  \kappa^2 )^{-1}(u-\bar u^\eps) \, dx, 
\end{equation}
for all $u \in \mathcal A$, where 
\begin{align}
  \label{A}
  \mathcal A := BV(\TT; \{-1, 1\}). 
\end{align}
We wish to understand the asymptotic properties of the energy $E^\eps$
in \eqref{E2} as $\eps \to 0$ when all other parameters are fixed. We
then relate our conclusions based on the study of this energy to its
diffuse interface version, which under the same scaling assumptions
takes the form
\begin{align}
  \label{EE2}
  \mathcal{E}^\eps[u] = \int_\TT \left( \frac{\eps^2}{2} |\nabla u|^2
    + W(u)\ + \frac{1}{2}(u- \bar u^\eps)(-\Delta) ^{-1} (u- \bar
    u^\eps) \right) dx,
\end{align}
with $u \in \mathcal A^\eps$, where
\begin{align}
  \label{Aeps}
  \mathcal A^\eps := \left\{ u \in H^1(\TT) : {1 \over \ell^2}
    \int_\TT u \, dx = \bar u^\eps \right\}.
\end{align}
Here the symmetric double-well potential $W \geq 0$ needs to satisfy
\begin{align}
  \label{W}
  W(1) = 0, \qquad W''(1) = {1 \over \kappa^2}, \qquad \int_{-1}^{1}
  \sqrt{2 W(u)} \, du = 1,
\end{align}
in order for $E^\eps$ to be compatible with $\mathcal E^\eps$ (see
further discussion at the beginning of Sec. \ref{sec:some-auxil-lemm}
and \cite[Sec. 4]{m:cmp10} for precise assumptions on $W$).  We note
that the relation between $E^\eps$ and $\mathcal E^\eps$ does not
amount to a straightforward application of the standard Modica-Mortola
argument \cite{modica77,modica87}, as will be explained in more detail
in Sec. \ref{sec:diff-interf-energy}. A formal application of the
latter to \eqref{EE2} would result in an energy of the type in
\eqref{E2}, but with the same (i.e., unscreened) Coulomb kernel as in
\eqref{EE2}, which is {\em not} $\Gamma$-equivalent to
$\mathcal{E}^\eps$. We also note that at the level of the energy
minimizers the relation between the two functionals was established in
\cite{m:cmp10}.

\subsection{Sharp interface energy}
\label{sec:sharp-interf-energy}

The sharp interface energy in \eqref{E2} is most conveniently
expressed in terms of {\em droplets}, i.e., the connected components
$\Omega_i^+$ of the set $\Omega^+ :=\{ u = +1\}$ (see Lemma
\ref{l:jord} for technical details).  Inserting
\begin{align}
  \label{chii}
  u = -1 + 2 \sum_i \chi_{\Omega_i^+},
\end{align}
into \eqref{E2}, where $\chi_{\Omega_i^+}$ are the characteristic
functions of $\Omega_i^+$, expressing the result via $G$ that solves
\begin{align}
  \label{G2}
  -\Delta G(x) + \kappa^2 G(x) = \delta(x) \qquad \text{in} \quad \TT,
\end{align}
expanding all the terms and using the fact that $\int_{\TT} G(x) dx =
\kappa^{-2}$, we arrive at (see also \cite{m:cmp10})
\begin{align}
  \label{EEOm}
  E^\eps[u] & = \frac{\ell^2 (1 + \bar u^\eps)^2}{2 \kappa^2} \notag \\
  & + \sum_i \Big\{ \eps |\partial \Omega_i^+| - 2 \kappa^{-2} (1 +
  \bar u^\eps) |\Omega_i^+| \Big\} + 2 \sum_{i,j} \int_{\Omega_i^+}
  \int_{\Omega_j^+} G(x - y) \, dx \, dy,
\end{align}
where we took into account the translational symmetry of the problem
in $\TT$. Moreover, since the optimal configurations for $\Omega_i^+$
are expected to consist of droplets of size of order $\eps^{1/3} |\ln
\eps|^{-1/3}$ (see \eqref{hier} and the discussion around), it is
convenient to introduce the rescaled area and perimeter of each
droplet:
\begin{align}
  \label{AP}
  A_i := \eps^{-2/3} |\ln \eps|^{2/3} |\Omega_i^+|, \qquad P_i :=
  \eps^{-1/3} |\ln \eps|^{1/3} |\partial \Omega_i^+|.
\end{align}
Similarly, let us introduce the suitably rescaled measure $\mu$
associated with the droplets:
\begin{align}
  \label{mu}
  d \mu(x) := \eps^{-2/3} |\ln \eps|^{-1/3} \sum_i
  \chi_{\Omega_i^+}(x) dx = \frac{1}{2} \eps^{-2/3} |\ln \eps|^{-1/3}(1+u)\, dx.
\end{align}
Note that by the definitions in \eqref{AP} and \eqref{mu} we have
\begin{align}
  \label{muA}
  {1 \over |\ln \eps|} \sum_i A_i = \int_\TT d \mu,
\end{align}
and the energy $E^\eps[u]$ may be rewritten as
\begin{align}
  E^\eps[u] = \eps^{4/3} |\ln \eps|^{2/3} \left( {\bar \delta^2 \ell^2
      \over 2 \kappa^2} + \bar E^\eps[u] \right), \label{EEE}
\end{align}
where
\begin{align}
  \bar E^\eps[u] := {1 \over |\ln \eps|} \sum_i \left( P_i - {2 \bar
      \delta \over \kappa^2} A_i \right) + 2 \int_\TT \int_\TT G(x -
  y) d \mu(x) d \mu(y). \label{Ebar}
\end{align}

We now state our $\Gamma$-convergence result, which is obtained for
configurations $(u^\eps)$ that obey the optimal energy scaling,
i.e. when $\bar E^\eps[u^\eps]$ remains bounded as $\eps \to 0$. The
result is obtained with the help of the framework established in
\cite{sandier00}, where an analogous result for the Ginzburg-Landau
functional of superconductivity was obtained. What we show is that the
limit functional $E^0$ depends only on the limit density $\mu$ of the
droplets (more precisely, on a limit measure $\mu \in \mathcal
M^+(\TT) \cap H^{-1}(\TT)$, see Lemma \ref{l:Hm1} for technical
details about such measures). In passing to the limit the second term
in \eqref{Ebar} remains unchanged, while the first term is converted
into a term proportional to the integral of the measure. The
proportionality constant is non-trivially determined by the optimal
droplet profile that will be discussed later on. We give the statement
of the result in terms of the original screened sharp interface energy
$E^\eps$, which is defined in terms of $u \in \mathcal A$.  In the
proof, we work instead with the equivalent energy $\bar E^\eps$, which
is defined through $\{A_i^\eps\}$, $\{P_i^\eps\}$ and $\mu^\eps$
corresponding to $u = u^\eps$ (cf. \eqref{EEE} and \eqref{Ebar}).

\begin{theorem} \textbf{\emph{($\Gamma$-convergence of
      $E^{\varepsilon}$)}}
  \label{main}
  \noindent Let $E^{\varepsilon}$ be defined by (\ref{E2}) with $\bar
  u^\eps$ given by \eqref{ubeps}. Then, as $\eps \to 0$ we have
  that \[ \eps^{-4/3} |\ln \eps|^{-2/3} E^{\varepsilon}
  \stackrel{\Gamma}\to E^0[\mu] := \frac{\bar \delta^2
    \ell^2}{2\kappa^2} + \left(3^{2/3} - \frac{2 \bar
      \delta}{\kappa^2} \right) \int_\TT d\mu + 2\int_\TT \int_\TT G(x
  - y) d \mu(x) d \mu(y),
\]
where $\mu \in \mathcal{M}^+(\TT) \cap H^{-1}(\TT)$.  More precisely,
we have
\begin{itemize} \item[i)] (Lower Bound) Let $(u^\eps) \in \mathcal A$
  be such that
  \begin{align}\label{ZO.1} \limsup_{\varepsilon \to 0}
    \eps^{-4/3} |\ln \eps|^{-2/3} E^{\varepsilon}[u^{\varepsilon}] <
    +\infty,
\end{align}
let 
\begin{align}
  \label{mueps}
  d \mu^\eps(x) := \tfrac12 \eps^{-2/3} |\ln \eps|^{-1/3} (1 +
  u^\eps(x)) dx,
\end{align}
and let $v^\eps$ satisfy
\begin{equation}\label{vdef}
  -\Delta v^{\varepsilon} + \kappa^2 v^{\varepsilon} =
  \mu^{\varepsilon} \qquad \text{in} \quad \TT. 
\end{equation}

Then, up to extraction of a subsequence, we have
\begin{align} 
  \mu^{\varepsilon} \rightharpoonup \mu \textrm{ in } (C(\TT))^*, \;\;
  v^{\varepsilon} \rightharpoonup v \textrm{ in } H^1(\TT), \nonumber
\end{align}
as $\eps \to 0$, where $\mu \in \mathcal{M}^+(\TT) \cap H^{-1}(\TT)$
and $v \in H^1(\TT)$ satisfy
\begin{equation}
  \label{PDE} -\Delta v + \kappa^2 v = \mu
  \qquad \text{in} \quad \TT.
\end{equation}
Moreover, we have
\[\liminf_{\varepsilon \to 0}  \eps^{-4/3} |\ln \eps|^{-2/3}
E^{\varepsilon}[u^{\varepsilon}] \geq E^0[\mu].\]

\item[ii)] (Upper Bound) Conversely, given $\mu \in \mathcal{M}^+(\TT)
  \cap H^{-1}(\TT)$ and $v \in H^1(\TT)$ solving \eqref{PDE}, there
  exist $(u^\eps) \in \mathcal A$ such that for the corresponding
  $\mu^{\varepsilon}$, $v^{\varepsilon}$ as in \eqref{mueps} and
  (\ref{vdef}) we have
\begin{align} 
  \mu^{\varepsilon} \rightharpoonup \mu \text{ in } (C(\TT))^*, \;\;
  v^{\varepsilon} &\rightharpoonup v \textrm{ in } H^1(\TT), \nonumber
\end{align}
as $\eps \to 0$, and
\[\limsup_{\varepsilon \to 0} \eps^{-4/3} |\ln \eps|^{-2/3}
E^{\varepsilon}[u^{\varepsilon}] \leq E^0[\mu].\]
\end{itemize}
\end{theorem}

We note that the limit energy $E^0$ obtained in Theorem \ref{main} may
be viewed as the {\em homogenized} (or mean-field) version of the
non-local part of the energy in the definition of $E^\eps$ associated
with the limit charge density $\mu$ of the droplets, plus a term
associated with the self-energy of the droplets. The functional $E^0$
is strictly convex, so there exists a unique minimizer $\bar \mu \in
\mathcal{M}^+(\TT) \cap H^{-1}(\TT)$ of $E^0$, which is easily seen to
be either $\bar\mu = 0$ for $\bar\delta \leq \tfrac12 3^{2/3}
\kappa^2$ or $\bar\mu = \tfrac12 (\bar\delta - \tfrac12 3^{2/3}
\kappa^2)$ otherwise. The latter can also be seen immediately from
Remark \ref{r:vloc} below, which gives a local characterization of the
limit energy $E^0$ (see Lemma \ref{l:Hm1}).

\begin{remark}
  \label{r:vloc} The limit energy $E^0$ in Theorem \ref{main} becomes
  local when written in terms of the limit potential $v$ defined in
  \eqref{PDE}:
  \begin{align}
    \label{Ev}
    E^0[\mu] = \frac{\bar \delta^2 \ell^2}{2\kappa^2} + \left(3^{2/3}
      \kappa^2 - 2 \bar \delta \right) \int_\TT v \, dx + 2\int_\TT
    \Big( |\nabla v|^2 + \kappa^2 v^2 \Big) dx.
  \end{align}
\end{remark}

Also, by the usual properties of $\Gamma$-convergence \cite{braides},
the optimal density $\bar \mu$ above is exhibited by the minimizers of
$E^\eps$ in the limit $\eps \to 0$, in agreement with \cite[Theorem
2.2]{m:cmp10}:

\begin{coro}
  \label{c:mubar}
  Let $\bar u^\eps$ be given by \eqref{ubeps} and let $(u^\eps) \in
  \mathcal A$ be minimizers of $E^\eps$ defined in \eqref{E2}. Then,
  letting $\bar \delta_c := \frac{1}{2} 3^{2/3} \kappa^2$, if
  $\mu^\eps$ is given by \eqref{mueps}, as $\eps \to 0$ we have
  \begin{itemize}
  \item[(i)] If $\bar\delta\le \bar\delta_c$, then
    \begin{align}
      \mu^\eps \rightharpoonup 0 \ \text{in } (C(\TT))^* \quad
      \text{and} \quad \eps^{-4/3} |\ln \eps|^{-2/3} \ell^{-2} \min
      E^\eps \to \frac{\bar \delta^2}{2\kappa^2}.
    \end{align}
  \item[(ii)] If $\bar\delta > \bar\delta_c$, then
    \begin{align}
      \mu^\eps \rightharpoonup \tfrac12 (\bar\delta - \bar \delta_c) \
      \text{in } (C(\TT))^* \quad \text{and} \quad \eps^{-4/3} |\ln
      \eps|^{-2/3} \ell^{-2} \min E^\eps \to \tfrac{\bar\delta_c}{2
        \kappa^2} (2 \bar\delta - \bar\delta_c).
    \end{align}
  \end{itemize}

  \end{coro}
\noindent In particular, since the minimal energy scales with the area
of $\TT$, it is an {\em extensive quantity}.

We next give the definition of almost minimizers with prescribed limit
density, for which a number of further results may be obtained.  These
can be viewed, e.g., as almost minimizers of $E^\eps$ in the presence
of an external potential. We note that in view of the strict convexity
of $E^0$, minimizing $E^0[\mu] + \int_\TT \varphi(x) d \mu(x)$ for a
given $\varphi \in H^1(\TT)$ one obtains a one-to-one correspondence
between the minimizing density $\mu$ and the potential $\varphi$. It
then makes sense to talk about almost minimizers of the energy
$E^\eps$ with prescribed limit density $\mu$ by viewing them as almost
minimizers of $E^\eps + \int_\TT \varphi^\eps d \mu^\eps$, where
$\varphi^\eps = \eps^{2/3} |\ln \eps|^{1/3} \varphi$.  Also, observe
that almost minimizers with the particular prescribed density $\bar
\mu$ from Corollary \ref{c:mubar} are simply almost minimizers of
$E^\eps$.  Below we give a precise definition.
\begin{definition}
  \label{amin}
  For a given $\mu \in \mathcal{M}^+(\TT) \cap H^{-1}(\TT)$, we will
  call every recovery sequence $( u^\eps ) \in \mathcal A$ in Theorem
  \ref{main}(ii) {\sl almost minimizers of $E^\eps$ with prescribed
    limit density $\mu$.}
\end{definition}

For almost minimizers with prescribed limit density, we show that in
the limit $\eps \to 0$ most of the droplets, with the exception of
possibly many tiny droplets comprising a vanishing fraction of the
total droplet area, converge to disks of radius $r = 3^{1/3}
\eps^{1/3} |\ln \eps|^{-1/3}$. More precisely, we have the following
result.

\begin{theorem}
  \label{almost}
  Let $(u^\eps) \in \mathcal A$ be a sequence of almost minimizers of
  $E^\eps$ with prescribed limit density $\mu$. For every $\gamma \in
  (0,1)$ define the set $I_\gamma^\eps := \{ i \in \mathbb N: 3^{2/3}
  \pi \gamma \leq A_i^\eps \leq 3^{2/3} \pi \gamma^{-1} \}$. Then
 \begin{align}
   \label{cor1.1} &\lim_{\varepsilon \to 0} \frac{1}{|\ln
     \varepsilon|}\sum_i \left( P_i^\eps - \sqrt
     {4 \pi A_i^\eps} \right) = 0,\\
   &\label{cor1.2bis} \lim_{\varepsilon \to 0} \frac{1}{|\ln
     \varepsilon|} \sum_{i \in I^\eps_\gamma}
   \left( A_i^\eps - 3^{2/3}\pi \right)^2 =0, \\
   &\lim_{\varepsilon \to 0} \frac{1}{|\ln \varepsilon|} \sum_{i \not
     \in I^\eps_\gamma} A_i^\eps = 0,\label{cor1.2}
 \end{align}
 where $\{A_i^\eps\}$ and $\{P_i^\eps\}$ are given by \eqref{AP} with
 $u = u^\eps$.
\end{theorem}
\noindent Note that we may use the isoperimetric deficit terms present
in \eqref{cor1.1} to control the Fraenkel asymmetry of the
droplets. The Fraenkel asymmetry measures the deviation of the set $E$
from the ball of the same area that best approximates $E$ and is
defined for any Borel set $E \subset \mathbb R^2$ by
\begin{align}
  \label{fraenkel}
  \alpha (E)= \min \frac{| E\triangle B|}{|E|},
\end{align}
where the minimum is taken over all balls $B \subset \mathbb R^2$ with
$|B|=|E|$, and $\triangle$ denotes the symmetric difference between
sets. Note that the following sharp quantitative isoperimetric
inequality holds for $\alpha(E)$ \cite{fusco08}:
\begin{align}
  \label{isodef}
  |\partial E|- \sqrt{4\pi |E|} \ge C \alpha^2 (E) \sqrt{|E|},
\end{align}
with some universal constant $C > 0$.  As a direct consequence of
Theorem \ref{almost} and \eqref{muA}, we then have the following
result.

\begin{coro}
  \label{c:Nlim}
  Under the assumptions of Theorem \ref{almost}, when $\int_\TT d \mu
  > 0$ we have
 \begin{align}
   \label{Nlim}
   \lim_{\eps \to 0} {3^{2/3} \pi |I_\gamma^\eps| \over |\ln \eps|} =
   \int_\TT d \mu, \qquad \lim_{\eps \to 0} \,
   \frac{1}{|I^\eps_\gamma|} \sum_{i\in I_\gamma^\ep}     \alpha
   (\Omega_i^+)= 0,
 \end{align}
 where $ |I_\gamma^\eps|$ denotes the cardinality of $I_\gamma^\eps$.
\end{coro}
\noindent This result generalizes the one in \cite{m:cmp10}, where it
was found that in the case of the minimizers {\em all} the droplets
are uniformly close to disks of the optimal radius $r = 3^{1/3}
\eps^{1/3} |\ln \eps|^{-1/3}$. What we showed here is that this result
holds for almost all droplets in the case of almost minimizers, in the
sense that in the limit almost all the mass concentrates in the
droplets of optimal area and vanishing isoperimetric deficit. We note
that the density $\mu$ is also the limit of the number density of the
droplets, up to a normalization constant, once the droplets of
vanishing area have been discarded.

The result that almost all droplets in almost minimizers with
prescribed limit density have asymptotically the {\em same} size, even
if the limit density is not constant in $\TT$ appears to be quite
surprising, since in this regime the self-interaction energy, which
governs the droplet shapes and partly their sizes is exactly of the
same order as the droplet mutual interaction energy, as was already
mentioned at the end of Sec. \ref{sec:intro}. In addition, the other
terms governing the droplets extracted in \eqref{Ebar} (the perimeter
and interaction with the background uniform charge) are equally
strong. This result would hold, for example, for minimizers of the
energy in the presence of a non-uniform potential, i.e., with a term
$\frac12 \eps^{2/3} |\ln \eps|^{1/3} \int_\TT \varphi(x) u(x) \, dx$
added to $E^\eps$ in \eqref{E2} (see also the paragraph before
Definition \ref{amin}). It means that while the density of the energy
minimizing droplets would be dependent on $\varphi$, their radii would
not. We note that this observation is consistent with the expectation
that quantum mechanical charged particle systems form Wigner crystals
at low particle densities \cite{lundqvist,wigner34,grimes79}. Let us
point out that the Ohta-Kawasaki energy $\mathcal E^\eps$ bears
resemblance with the classical Thomas-Fermi-Dirac-Von Weizs\"acker
model arising in the context of density functional theory of quantum
systems (see e.g. \cite{lebris05,lieb81,lundqvist}).

\subsection{Diffuse interface energy}
\label{sec:diff-interf-energy}
  
We now turn to relating the results obtained so far for the screened
sharp interface energy $E^\eps$ to the original diffuse interface
energy $\mathcal E^\eps$. On the level of the minimal energy, the
asymptotic equivalence of the energies in the considered regime,
namely, that for every $\delta > 0$
\begin{align}
  \label{aseq}
  (1 - \delta) \min E^\eps \leq \min \mathcal E^\eps \leq (1 + \delta)
  \min E^\eps
\end{align}
for $ \eps \ll 1$ was established in \cite[Theorem 2.3]{m:cmp10}. The
main idea of the proof in \cite{m:cmp10} is for a given function
$u^\eps \in \mathcal A^\eps$ to establish an approximate lower bound
for $\mathcal E^\eps[u^\eps]$ in terms of $(1 - \delta) E^\eps[\tilde
u^\eps]$ for some $\tilde u^\eps \in \mathcal A$, with $\delta > 0$
which can be chosen arbitrarily small for $\eps \ll 1$. The matching
approximate upper bound is then obtained by a suitable lifting of the
minimizer $u^\eps \in \mathcal A$ of $E^\eps$ into $\mathcal A^\eps$.

\begin{figure}
  \centering
  \includegraphics[width=8cm]{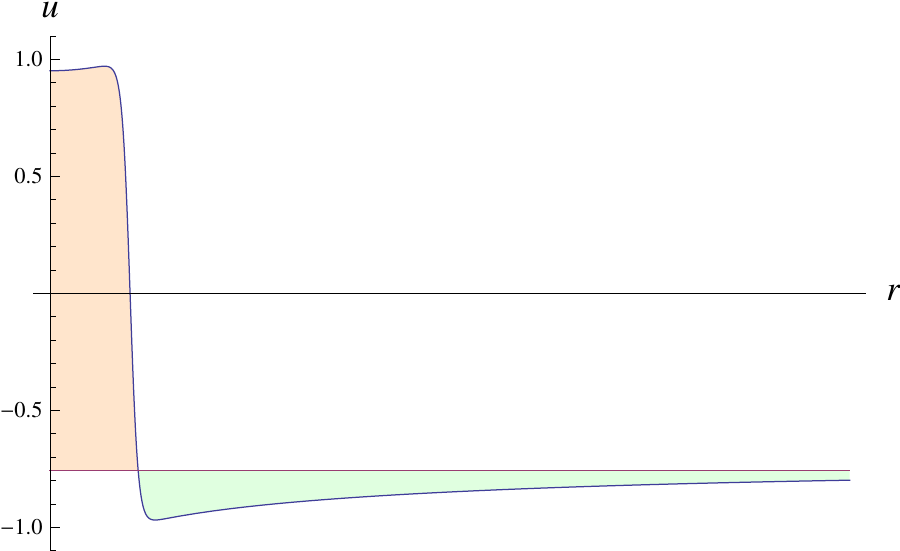}  
  \caption{A qualitative form of the $u$-profile for a single droplet
    from the Euler-Lagrange equation associated with $\mathcal E$. The
    horizontal line shows the level corresponding to $\bar u$. Charge
    is transferred from the region where $u < \bar u$ (depletion shown
    in green) to the region where $u > \bar u$ (excess shown in
    orange). At the sharp interface level the corresponding profile is
    given by $\text{sgn} (u)$, whose average charge is {\em not} equal
    to $\bar u$.}
  \label{fig:uvr}
\end{figure}

Here we show that the procedure outlined above may also be applied to
almost minimizers of $\mathcal E^\eps$ in a suitably modified version
of Definition \ref{amin} involving $\mathcal E^\eps$, using almost
minimizers of $E^\eps$ for comparisons. We note right away, however,
that it is not possible to simply replace $E^\eps$ with $\mathcal
E^\eps$ in Definition \ref{amin}. The reason for this is the presence
of the mass constraint in the definition of the admissible class
$\mathcal A^\eps$ for $\mathcal E^\eps$. This implies, for example,
that any sequence of almost minimizers $(u^\eps) \in \mathcal A^\eps$
of $\mathcal E^\eps$ must satisfy $\ell^{-2} \int_\TT d \mu^\eps =
\tfrac12 \bar\delta$, while, according to Corollary \ref{c:mubar}, for
sequences of almost minimizers $(u^\eps) \in \mathcal A$ of $E^\eps$
we have $\ell^{-2} \int_\TT d \mu^\eps \to \bar \mu \not= \tfrac12
\bar\delta$. This phenomenon is intimately related to the effect of
screening of the Coulombic potential from the droplets by the
compensating charges that move into their vicinity \cite{m:pre02}. For
a single radially symmetric droplet the solution of the Euler-Lagrange
equation associated with $\mathcal E^\eps$ has the form shown in
Fig. \ref{fig:uvr}, which illustrates the gap between the
``prescribed'' total charge at the diffuse interface level and the
total charge at the sharp interface level.

In order to be able to extract the limit behavior of the energy, we
need to take into consideration the redistribution of charge discussed
above and define almost minimizers with prescribed limit density that
belong to $\mathcal A^\eps$ and for which the screening charges are
removed from the consideration of convergence to the limit
density. Hence, given a candidate function $u^\eps \in \mathcal
A^\eps$, we define a new function
\begin{align}
  \label{ueps0}
  u^\eps_0(x) :=
  \begin{cases}
    +1, & u^\eps(x) > 0, \\
    -1, & u^\eps(x) \leq 0,
  \end{cases}
\end{align} 
whose jump set coincides with the zero level set of $u^\eps$. This
introduces a nonlinear filtering operation that eliminates the effect
of the small deviations of $u^\eps$ from $\pm 1$ in almost minimizers
on the limit density (compare also with \cite{km:jns11}). The measure
$\mu_0^\eps$ associated with the droplets is now defined via
\begin{align}
  \label{mu0eps}
  d \mu_0^\eps := \tfrac12 \eps^{-2/3} |\ln \eps|^{-1/3} (1 +
  u_0^\eps(x)) dx.
\end{align}

We can follow the ideas of \cite{m:cmp10} to establish an analog of
Theorem \ref{main} for the diffuse interface energy. To avoid many
technical assumptions, we formulate the result for a specific choice
of $W(u) = \frac{9}{32} (1 - u^2)^2$ and $\kappa = 1 / \sqrt{W''(1)} =
\frac23$ (see the discussion at the beginning of
Sec. \ref{sec:some-auxil-lemm}). A general result may easily be
reconstructed. Also, we make a technical assumption to avoid dealing
with the case $\limsup_{\eps \to 0} \|u^\eps\|_{L^\infty(\TT)} > 1$,
when spiky configurations in which $|u^\eps|$ significantly exceeds 1
in regions of vanishing size may appear. We note that this condition
is satisfied by the minimizers of $\mathcal E^\eps$ \cite[Proposition
4.1]{m:cmp10}.

\begin{theorem}
\textbf{\emph{($\Gamma$-convergence of
      $\mathcal E^{\varepsilon}$)}}
  \label{main2}
  \noindent Let $\mathcal E^{\varepsilon}$ be defined by (\ref{EE2})
  with $W(u) = \tfrac{9}{32} (1 - u^2)^2$ and $\bar u^\eps$ given by
  \eqref{ubeps}. Then, as $\eps \to 0$ we have that \[ \eps^{-4/3}
  |\ln \eps|^{-2/3} \mathcal E^{\varepsilon} \stackrel{\Gamma}\to
  E^0[\mu] := \frac{\bar \delta^2 \ell^2}{2\kappa^2} + \left(3^{2/3} -
    \frac{2 \bar \delta}{\kappa^2} \right) \int_\TT d\mu + 2\int_\TT
  \int_\TT G(x - y) d \mu(x) d \mu(y),
\]
where $\mu \in \mathcal{M}^+(\TT) \cap H^{-1}(\TT)$ and $\kappa =
\tfrac23$.  More precisely, we have
\begin{itemize} \item[i)] (Lower Bound) Let $(u^\eps) \in \mathcal
  A^\eps$ be such that $\limsup_{\eps \to 0}
  \|u^\eps\|_{L^\infty(\TT)} \leq 1$ and
  \begin{align}\label{2ZO.1} \limsup_{\varepsilon \to 0}
    \eps^{-4/3} |\ln \eps|^{-2/3} \mathcal
    E^{\varepsilon}[u^{\varepsilon}] < +\infty,
\end{align}
and let $\mu_0^\eps(x)$ be defined by \eqref{ueps0} and
\eqref{mu0eps}.

Then, up to extraction of subsequences, we have
\begin{align} 
  \mu_0^{\varepsilon} \rightharpoonup \mu \textrm{ in }
  (C(\TT))^*,\nonumber
\end{align}
as $\eps \to 0$, where $\mu \in \mathcal{M}^+(\TT) \cap H^{-1}(\TT)$.
Moreover, we have $\limsup_{\eps \to 0} \|u^\eps\|_{L^\infty(\TT)} =
1$ and
\[\liminf_{\varepsilon \to 0}  \eps^{-4/3} |\ln \eps|^{-2/3}
\mathcal E^{\varepsilon}[u^{\varepsilon}] \geq E^0[\mu].\]

\item[ii)] (Upper Bound) Conversely, given $\mu \in \mathcal{M}^+(\TT)
  \cap H^{-1}(\TT)$, there exist $(u^\eps) \in \mathcal A^\eps$ such
  that $\limsup_{\eps \to 0} \|u^\eps\|_{L^\infty(\TT)} = 1$ and for
  $\mu_0^{\varepsilon}$ defined by \eqref{ueps0} and \eqref{mu0eps} we
  have
\begin{align} 
  \mu_0^{\varepsilon} \rightharpoonup \mu \text{ in } (C(\TT))^*,
  \nonumber
\end{align}
as $\eps \to 0$, and
\[\limsup_{\varepsilon \to 0} \eps^{-4/3} |\ln \eps|^{-2/3}
\mathcal E^{\varepsilon}[u^{\varepsilon}] \leq E^0[\mu].\]
\end{itemize}
\end{theorem}

Based on the result of Theorem \ref{main2}, we have the following
analog of Corollary \ref{c:mubar} for the diffuse interface energy
$\mathcal E^\eps$.

\begin{coro}
  \label{c:mubar2}
  Let $\bar u^\eps$ be given by \eqref{ubeps} and let $(u^\eps) \in
  \mathcal A^\eps$ be minimizers of $\mathcal E^\eps$ defined in
  \eqref{EE2} with $W(u) = \tfrac{9}{32} (1 - u^2)^2$. Then, letting
  $\kappa = \tfrac23$ and $\bar \delta_c := \frac{1}{2} 3^{2/3}
  \kappa^2$, if $u^\eps_0$ and $\mu_0^\eps$ are defined via
  \eqref{ueps0} and \eqref{mu0eps}, respectively, as $\eps \to 0$ we
  have
  \begin{itemize}
  \item[(i)] If $\bar\delta\le \bar\delta_c$, then
    \begin{align}
      \mu_0^\eps \rightharpoonup 0 \ \text{in } (C(\TT))^*, \quad
      \text{and} \quad \eps^{-4/3} |\ln \eps|^{-2/3} \ell^{-2} \min
      E^\eps \to \frac{\bar \delta^2}{2\kappa^2}.
    \end{align}
  \item[(ii)] If $\bar\delta > \bar\delta_c$, then
    \begin{align}
      \mu^\eps \rightharpoonup \tfrac12 (\bar\delta - \bar \delta_c) \
      \text{in } (C(\TT))^*, \quad \text{and} \quad \eps^{-4/3} |\ln
      \eps|^{-2/3} \ell^{-2} \min \mathcal E^\eps \to
      \tfrac{\bar\delta_c}{2 \kappa^2} (2 \bar\delta - \bar\delta_c).
    \end{align}
  \end{itemize}
\end{coro}

\noindent In addition, we have the following analog of Theorem
\ref{almost}, which, in particular, applies to minimizers of the
diffuse interface energy $\mathcal E^\eps$.
\begin{theorem}
  \label{almost2}
  Let $(u^\eps) \in \mathcal A^\eps$ be a recovery sequence as in
  Theorem \ref{main2}(ii) and let $\int_\TT d \mu > 0$. Then there
  exists a set of finite perimeter $\Omega^+$ such that if
  $\Omega^+_i$ are its connected components, then the conclusion of
  Theorem \ref{almost} holds with $\{A_i^\eps\}$ and $\{P_i^\eps\}$
  given by \eqref{AP}, and
  \begin{align}
    \label{limOmpu}
    \lim_{\eps \to 0} { | \Omega^+ \triangle \{ u^\eps > 0 \} | \over
      |\Omega^+|} = 0.
  \end{align}
\end{theorem}

Theorem \ref{almost2} essentially says that the zero superlevel set of
$u^\eps$ from every recovery sequence of Theorem \ref{main2} may be
well approximated in $L^1$ sense by a union of of droplets that are,
in turn, close to disks of radius $r = 3^{1/3} \eps^{1/3} |\ln
\eps|^{-1/3}$ for $\eps \ll 1$. The $L^1$ error arises, since we do
not have control on the perimeter of every superlevel set of
$u^\eps$. At the same time, the choice of the zero superlevel set of
$u^\eps$ in the definition of the truncated version $u^\eps_0$ of
$u^\eps$ in \eqref{ueps0} was arbitrary. We could equivalently use the
superlevel set $\{u^\eps > c \}$ for any $c \in (-1,1)$ fixed. Also,
we point out that the conclusions of Corollary \ref{c:Nlim} remain
true for $\Omega^+$ in Theorem \ref{almost2} under the assumptions of
Theorem \ref{main2}. 

\section{Some auxiliary lemmas}
\label{sec:some-auxil-lemm}

In this section we collect some technical results that are needed in
the proofs of our theorems. Before proceeding to those results,
however, let us first show that the assumption in \eqref{W} that needs
to be imposed on $W$ in order to have $\Gamma$-equivalence between
$E^\eps$ and $\mathcal E^\eps$ defined in \eqref{E2} and \eqref{EE2},
respectively, and, hence, the conclusion of Theorem \ref{main2} (see
also \cite{m:cmp10}), is not restrictive. Indeed, given the definition
of $\mathcal E^\eps$ in \eqref{EE2}, introduce a rescaling:
\begin{align}
  \label{EEEre}
  \qquad W = \lambda^2 \widetilde W, \qquad \ell = \lambda \tilde
  \ell, \qquad \eps = \lambda^2 \tilde\eps.
\end{align}
Then it is easy to see that if $\tilde u(x) := u(\lambda x)$, then
$\mathcal E^\eps[u] = \lambda^4 \tilde{\mathcal E}^{\tilde\eps}[\tilde
u]$, where $\tilde{\mathcal E}^{\tilde\eps}$ is obtained from
\eqref{EE2} by replacing all the quantities with their tilde
equivalents. In particular, choosing $\lambda = 3 / (2 \sqrt{2})$ we
can relate the original Ohta-Kawasaki energy $\tilde{\mathcal
  E}^{\tilde\eps}$, which has $\widetilde W(u) = \frac14 (1 - u^2)^2$
\cite{ohta86}, to the energy appearing in the statement of Theorem
\ref{main2}. The choice of $W$ satisfying \eqref{W} simply avoids many
extra constants appearing in the statements of results.

As was already mentioned, the energy $E^\eps$ may be alternatively
written in terms of the level sets of $u$.  Indeed, when $E^\eps[u] <
+\infty$, the set $\Omega^+ := \{ u = +1 \}$ is a set of finite
perimeter (for precise definitions and the terminology used below, see
\cite{ambrosio01}).  We then have the following result about
decomposing $\Omega^+$ into measure theoretic connected components
$\Omega_i^+$, which in view of the scaling of the upper bound on
energy will be shown to hold for all sufficiently small $\eps > 0$.
Note that the latter assumption implies that each connected component
on the torus has the same geometric structure as connected components
of sets of finite perimeter in the whole plane, thus excluding a
possibility of stripe-like components winding around the torus and,
hence, justifying the use of the word ``droplet''.  We will also make
repeated use of the basic fact that the diameter of a connected
component is essentially controlled by its perimeter (i.e., modulo a
set of measure zero).

\begin{lem}
  \label{l:jord}
  Let $\Omega^+ \subset \TT$ be a set of finite perimeter, and assume
  that $|\Omega^+| \leq \frac{1}{160} \ell^2$ and $|\partial \Omega^+|
  \leq \frac{1}{10} \ell$. Then $\Omega^+$ may be uniquely decomposed
  (up to negligible sets) into an at most countable union of connected
  sets $\Omega_i^+$ of positive measure, which, after a suitable
  translation and extension to $\mathbb R^2$, are essentially bounded
  and whose essential boundaries $\partial^M \Omega_i^+$ are (up to
  negligible sets) at most countable unions of Jordan curves that are
  essentially disjoint. Furthermore, we have
\begin{equation}\label{diamOmi}
  \mathrm{ess} \, \diam  \Omega_i^+ \, \leq \frac12
  |\partial \Omega_i^+|. 
\end{equation}
\end{lem}

\begin{proof}
  Let $\Omega^+_\#$ be the periodic extension of $\Omega^+$ from $\TT$
  to $\mathbb R^2$, and let $K_R := (-R, R)^2$. Then for every $R \in
  (\ell, \tfrac32 \ell)$ the set $\Omega^+_\# \cap K_R \subset \mathbb
  R^2$ is a set of finite perimeter, and we have
  \begin{align}
    \label{Pomper}
    |\partial (\Omega^+_\# \cap K_R)| \leq 9 |\partial \Omega^+| +
    \mathcal H^1(\mathring{\Omega}^+_\# \cap \partial K_R).
  \end{align}
On the other hand, by  the co-area formula we have 
\begin{align}
  \label{coarea}
  \int_\ell^{\frac32 \ell} \mathcal H^1(\mathring{\Omega}^+_\#
  \cap \partial K_t) dt = |\Omega^+_\# \cap K_{\frac32 \ell}
  \backslash K_\ell| \leq 8 |\Omega^+|.
\end{align}
Therefore, there exists $R \in (\ell, \tfrac32 \ell)$ such that
$\mathcal H^1(\mathring{\Omega}^+_\# \cap \partial K_R) \leq 16
\ell^{-1} |\Omega^+|$. Using the assumptions of the Lemma, we then
conclude that $\mathcal H^1(\mathring{\Omega}^+_\# \cap \partial K_R)
\leq \tfrac{1}{10} \ell$ and by \eqref{Pomper} we have $|\partial
(\Omega^+_\# \cap K_R)| \leq \ell$.

We now apply the results of \cite[Corollary 1 and Theorem
8]{ambrosio01} to the set $\Omega^+_\# \cap K_R$ to obtain its
decomposition into connected components and denote by $\Omega_i^+$
those components for which $|\Omega_i^+ \cap K_{\frac12 \ell}| >
0$. In turn, by \cite[Theorem 7 and Lemma 4]{ambrosio01} and noting
that in view of \cite[Proposition 6(ii)]{ambrosio01} it is sufficient
to consider only simple sets (see \cite[Definition 3]{ambrosio01}), we
have that $\Omega_i^+$ satisfy \eqref{diamOmi}. Therefore, from our
estimate on $|\partial (\Omega^+_\# \cap K_R)|$ we conclude that
$|\Omega_i^+ \cap K_{\frac32 \ell} \backslash K_\ell| = 0$, and so
$|\partial \Omega_i^+|$ does not have contributions from $\partial
K_R$.  Together with the assumptions of the Lemma, this then implies
that each $\Omega_i^+$ is essentially contained, after a suitable
translation, in $K_{\frac14 \ell}$. Finally, identifying all
translates of $\Omega_i^+$ by $\pm \ell$ in either coordinate
direction with the connected components of $\Omega^+$ in $\TT$, we
obtain the desired decomposition of $\Omega^+ \subset \TT$ for which
\eqref{diamOmi} also holds in the case of the perimeter relative to
$\TT$.
\end{proof}

In the context of $\Gamma$-convergence the sets $\Omega_i^+$ may be
viewed as a suitable generalization of the droplets introduced earlier
in the studies of energy minimizing patterns \cite{m:cmp10}. Note,
however, that the sets $\Omega_i^+$ lack the regularity properties of
the energy minimizers in \cite{m:cmp10} and may in general be fairly
ill-behaved (in particular, they do not have to be simply
connected). Nevertheless, they are fundamental for the description of
the low energy states associated with $E^\eps$ and, in particular,
will be shown to be close, in some average sense, to disks of
prescribed radii for almost minimizers of energy.

We now discuss the precise nature of the limit measures appearing in
our analysis. We say that $\mu \in \mathcal M^+(\TT) \cap
H^{-1}(\TT)$, if the non-negative Radon measure $\mu$ has bounded
Coulombic energy, i.e., if
\begin{align}
  \label{dmubdG}
  \int_\TT \int_\TT G(x - y) \, d \mu(x) \, d \mu(y) < \infty.
\end{align}
Our notation is justified by the following fundamental properties of
such measures.

\begin{lem}
  \label{l:Hm1}
  Let $\mu \in \mathcal M^+(\TT)$ and let \eqref{dmubdG} hold. Then
  \begin{itemize}
  \item[(i)] $\mu$ can be extended to a bounded linear functional over
    $H^1(\TT)$.
  \item[(ii)] If
    \begin{align}
      \label{vGconv}
      v(x) := \int_\TT G(x - y) \, d \mu(y),
    \end{align}
    then $v \in H^1(\TT)$. Furthermore, $v$ solves
    \begin{align}
      \label{vmu}
      -\Delta v + \kappa^2 v = \mu,
    \end{align}
    weakly in $H^1(\TT)$, and
    \begin{align}
      \label{dvmu}
      \nabla v(x) = \int_\TT \nabla G(x - y) \, d \mu(y),
    \end{align}
    in the sense of distributions. 
  \item[(iii)] If $v$ is as in (ii), we have $\kappa^2 \int_\TT v \,
    dx = \int_\TT d \mu$ and
    \begin{align}
      \label{Gmumu}
      \int_\TT \int_\TT G(x - y) \, d \mu(x) \, d \mu(y) = \int_\TT
      \left( |\nabla v|^2 + \kappa^2 v^2 \right) dx.
    \end{align}
  \end{itemize}
\end{lem}

\begin{proof}
  We first show that $v$ defined in \eqref{vGconv} has distributional
  first derivatives in $L^2(\TT)$. Introduce $H(x) > 0$ defined for
  all $x \in \TT$ by
  \begin{align}
    \label{Hsum}
    H(x) := {1 \over 2 \pi} \sum_{\mathbf n \in \mathbb Z^2}
    {e^{-\kappa | x - \mathbf n \ell|} \over |x - \mathbf n \ell|},
  \end{align}
  whose Fourier coefficients are easily seen to be
  \begin{align}
    \label{Hk}
    \widehat H(k) := \int_\TT e^{i k \cdot x} H(x) \, dx = {1 \over
      \sqrt{\kappa^2 + |k|^2}}, \qquad \qquad k \in 2 \pi \ell^{-1}
    \mathbb Z^2.
  \end{align}
  Indeed, $H(x)$ may be viewed as the trace $\widetilde H(x, 0)$ of
  the solution of
  \begin{align}
    \label{H3}
    -\Delta \widetilde H(x) + \kappa^2 \widetilde H(x) = 2
    \sum_{\mathbf n \in \mathbb Z^2 \times \{0\} } \delta(x - \mathbf
    n \ell), \qquad \qquad x \in \mathbb R^3,
  \end{align}
  which is given by the same formula as in \eqref{Hsum}. Denoting by
  $\widetilde H_k(z)$ the Fourier coefficients of $\widetilde H(x, z)$
  in $x \in \TT$, from \eqref{H3} one obtains that $\widetilde H_k(z)$
  solves
  \begin{align}
    \label{Hz}
    - \widetilde H_k''(z) + (\kappa^2 + |k|^2) \widetilde H_k(z) = 2
    \delta(z), 
  \end{align}
  whose explicit solution is $\widetilde H_k(z) = e^{-z \sqrt{\kappa^2
      + |k|^2}} / \sqrt{\kappa^2 + |k|^2}$.

  From \eqref{Hk} and the equation satisfied by $G$ one immediately
  concludes that
  \begin{align}
    \label{GH}
    G(x) = \int_\TT H(x - y) H(y) \, dy.
  \end{align}
  Furthermore, by direct inspection one can see that
  \begin{align}
    \label{GHbound}
    |\nabla G(x)| \leq C H(x) \qquad \qquad \forall x \in \TT,
  \end{align}
  for some $C > 0$. In addition, defining 
  \begin{align}
    \label{b}
    b(x) := \int_\TT H(x - y) \, d \mu(y),
  \end{align}
  by Tonelli's theorem and \eqref{GH} we have
  \begin{align}
    \label{bH}
    \int_\TT b^2 dx = \int_\TT \int_\TT \int_\TT H(x - z) H(y - z)
    \, d \mu(x) \, d \mu(y) \, dz \notag \\
    = \int_\TT \int_\TT G(x - y) \, d \mu(x) \, d \mu(y),
  \end{align}
  and, hence, by \eqref{dmubdG} we
  have $b \in L^2(\TT)$. Therefore, if
  \begin{align}
    \label{hG}
    h(x) := \int_\TT \nabla G(x - y) \, d \mu(y),
  \end{align}
  then by \eqref{GHbound} and \eqref{b} it is well defined, and we
  have $h \in L^2(\TT; \mathbb R^2)$ as well.

  Now, testing \eqref{vGconv} with $\nabla \varphi$, where $\varphi
  \in C^\infty(\TT)$, yields
  \begin{align}
    \label{nabphiv}
    -\int_\TT \nabla \varphi(x) v(x) \, dx = - \int_\TT \int_\TT
    \nabla \varphi(x) G(x - y) d \mu(y) = \int_\TT \varphi(x) h(x) \,
    dx,
  \end{align}
  which is justified by Fubini's theorem, in view of the fact that $h
  \in L^2(\TT; \mathbb R^2)$. Hence $\nabla v = h \in L^2(\TT; \mathbb
  R^2)$ distributionally, proving \eqref{dvmu}. To prove that $v \in
  H^1(\TT)$, observe that by Tonelli's theorem
  \begin{align}
    \label{vGG}
    \int_\TT v^2 dx = \int_\TT \int_\TT \int_\TT G(x - z) G(y - z) \,
    d \mu(x) \, d \mu(y) \, dz \leq C \left( \int_\TT d \mu \right)^2, 
  \end{align}
  for some $C > 0$. On the other hand, since by maximum principle
  $G(x) \geq c > 0$ for all $x \in \TT$, we conclude that
  \begin{align}
    \label{Gmumin}
    c \left( \int_\TT d \mu \right)^2 \leq \int_\TT \int_\TT G(x - y)
    \, d \mu(x) \, d \mu(y).
  \end{align}
  Therefore, by \eqref{dmubdG} we have that $\mu$ is bounded in the
  sense of measures, and so from \eqref{vGG} follows that $v \in
  L^2(\TT)$ as well.

We may next show that \eqref{vmu} holds distributionally by testing
$v$ in \eqref{vGconv} with $-\Delta \varphi + \kappa^2 \varphi \in
C^\infty(\TT)$ and integrating by parts.  Then, to conclude the proof
of the lemma, we test \eqref{vmu} with $\varphi \in C^\infty(\TT)$ and
apply the Cauchy-Schwarz inequality to obtain
  \begin{align}
    \label{vphimu}
    \left| \int_\TT \varphi \, d \mu \right| = \left| \int_\TT \left( \nabla
      \varphi \cdot \nabla v + \kappa^2 \varphi v \right) dx \right|
  \leq C \| v \|_{H^1(\TT)} \| \varphi \|_{H^1(\TT)}, 
  \end{align}
  for some $C > 0$. This yields (i), and, hence, \eqref{vmu} also
  holds weakly in $H^1(\TT)$. Finally, to obtain (iii), we interpret
  $\mu$ in \eqref{vmu} as an element of $H^{-1}(\TT)$ and test
  \eqref{vmu} with either 1 or $v$ itself.
\end{proof}

\begin{remark}
  It is not difficult to extend the proof of Lemma \ref{l:Hm1} to the
  case of measures with finite Coulombic energy defined on a
  sufficiently regular domain $\Omega$ with either Dirichlet or
  Neumann boundary conditions for the potential. In this case the role
  of $H$ would be played by the kernel of the Neumann-to-Dirichlet map
  for the operator $-\Delta + \kappa^2$ extended to $\Omega \times
  \mathbb R^+$.
\end{remark}

Observe that for the nontrivial minimizers we know from \cite{m:cmp10}
that $\bar E^\eps = O(1)$, $A_i = O(1)$ and $P_i = O(1)$ (and even
more precisely $A_i \simeq 3^{2/3} \pi $ and $P_i \simeq 2 \cdot
3^{1/3} \pi$), the number of droplets is $N = O(|\ln \varepsilon|)$,
and $\mu$ closely approximates the sum of Dirac masses at the droplet
centers with weights of order $|\ln \eps|^{-1}$. If, on the other
hand, the considered configurations only obey an energy bound under
the optimal scaling, then the same estimates turn out to hold for the
droplets on average. The precise result is stated in the following
lemma.

\begin{lem}
  \label{l:APmu}
  Let $(u^\eps) \in \mathcal A$, let $\limsup_{\varepsilon \to 0} \bar
  E^\eps[u^\eps] < +\infty$, and let $\{A_i^\eps\}$, $\{P_i^\eps\}$
  and $\mu^\eps$ be given by \eqref{AP} and \eqref{mu} with $u =
  u^\eps$. Then
  \begin{align} 
    \label{ZO.2} \limsup_{\varepsilon \to 0} \frac{1}{|\ln
      \varepsilon|} \sum_{i} P_i^\eps < +\infty, \qquad
    \limsup_{\varepsilon \to 0} \frac{1}{|\ln \varepsilon|} \sum_{i}
    A_i^\eps < +\infty, 
  \end{align}
  and 
  \begin{align}
    \label{mulinf}
    \limsup_{\varepsilon \to 0}  \int_\TT d \mu^\eps < +\infty.
  \end{align}
\end{lem}

\begin{proof}
  By \eqref{muA} and the positivity of $P_i^\eps$, we obtain the
  result, once we prove \eqref{mulinf}. To prove the latter, we simply
  note that if $\int_\TT d \mu^\eps \geq 2 \bar\delta / ( c
  \kappa^2)$, where $c$ is the same as in \eqref{Gmumin}, then by
  \eqref{muA} we have from the definition of $\bar E^\eps$ in
  \eqref{Ebar}:
  \begin{align}
    \label{GG0}
    \bar E^\ep[u] \geq - \frac{2\bar{\delta} }{\kappa^2 } \int_\TT d
    \mu^\eps + 2 c \left( \int_\TT d \mu^\eps \right)^2 \geq c \left(
      \int_\TT d \mu^\eps \right)^2,
  \end{align}
  which yields \eqref{mulinf}.
\end{proof}

\section{Proof of Theorem \ref{main}}
\label{sec:main}

Throughout all the proofs below, the values of $A_i^\eps$ and
$P_i^\eps$ are always the rescaled areas and perimeters, defined in
\eqref{AP}, of the connected components $\Omega_i^+$ of $\Omega^+ = \{
u = +1 \}$ for a given $u = u^\eps$, as in Lemma \ref{l:APmu}.  The
presentation is clarified by working with the rescaled energy $\bar
E^{\eps}$ defined by \eqref{Ebar} rather than $E^{\eps}$ directly. We
begin by proving Part i) of Theorem \ref{main}, the lower bound.

\subsection{Proof of lower bound, Theorem \ref{main} i)}

\bigskip

\noindent {\em Step 1: Estimate of $\bar E^{\varepsilon}$ in terms of
  $A_i^\eps$ and $P_i^\eps$.
}\\

\noindent First, for a fixed $\gamma \in (0,1)$ we define a
\emph{truncated} rescaled droplet area:
\begin{equation} \label{P1.N34} \tilde A_i^\eps := 
  \begin{cases}
    A_i^\eps, & \text{if } \ A_i^\eps <  3^{2/3} \pi \gamma^{-1} \, \\
    (3^{2/3} \pi \gamma^{-1})^{1/2} |A_i^\eps|^{1/2} &\text{if } \
    A_i^\eps \geq 3^{2/3} \pi \gamma^{-1},
  \end{cases}
\end{equation} 
and the isoperimetric deficit
\begin{equation}\label{isodefzero}
  I_\mathrm{def}^{\varepsilon} := \frac{1}{|\ln \varepsilon|} \sum_{i}
  \left(  P_i^\eps - \sqrt{4 \pi   A_i^\eps} \, \right) \geq 0,
\end{equation}
which will be used throughout the proof. The purpose of defining the
truncated droplet area in (\ref{P1.N34}) will become clear later. 
 
We start by writing $ \mu^\eps = \sum_i \mu_i^\eps$, with
 \begin{align} 
   \label{muieps}
   d \mu_i^{\varepsilon}(x) := \eps^{-2/3} |\ln \eps|^{-1/3}
   \chi_{\Omega_i^+} (x) dx, 
\end{align}
where $\Omega_i^+$ are the connected components of $\Omega^+ =
\{u^\eps = +1\}$, and the index $\eps$ was omitted from $\Omega_i^+$
to avoid cumbersome notation.  For small enough $\eps$ this is
justified by Lemma \ref{l:jord}, in view of the fact that for some $C
> 0$ we have
\begin{align}
  \label{pOmell}
  |\partial \Omega^+| \leq \eps^{-1} E^\eps[u^\eps]\leq C \eps^{1/3}
  |\ln \eps|^{2/3} ,
\end{align}
so $|\partial \Omega_i^+| \ll \ell$ whenever $\eps \ll 1$. In
particular, \eqref{diamOmi} holds for $\Omega_i^+$ when $\eps$ is
sufficiently small.

For a fixed $\rho > 0$ we introduce the ``far field truncation''
$G_\rho \in C^\infty(\TT)$ of the Green's function $G$:
\begin{align}
  \label{Geta}
  G_\rho (x - y) := G(x - y) \phi_\rho(|x - y|) \qquad \forall (x, y)
  \in \TT \times \TT,
\end{align}
where $\phi_\rho \in C^\infty(\mathbb R)$ is a monotonically
increasing cutoff function such that $\phi_\rho(t) = 0$ for all $t <
\frac12 \rho$ and $\phi_\rho(t) = 1$ for all $t > \rho$.  Then, for
sufficiently small $\eps$ we have $|\partial \Omega_i^+| \leq \rho$ in
view of \eqref{pOmell}, and from \eqref{Ebar} and \eqref{diamOmi} we
obtain
\begin{align} 
  \bar E^{\varepsilon}[u^{\varepsilon}] \geq
  &I_\mathrm{def}^{\varepsilon} + \frac{1}{|\ln \varepsilon|}
  \left(\sum_i \sqrt{4 \pi A_i^\eps} - \frac{2\bar \delta}{\kappa^2}
    A_i^\eps\right) \nonumber \\ &+ 2\sum_i \iint
  G(x-y)d\mu_i^{\varepsilon}(x)d\mu_i^{\varepsilon}(y)  \label{EGrho}\\
  &+ 2\iint G_\rho(x-y)d\mu^{\varepsilon}(x) d\mu^{\varepsilon}(y),
  \nonumber
\end{align}
where we used \eqref{diamOmi} and the positivity of $G$
(cf. e.g. \cite{m:cmp10}), and here and everywhere below we omit $\TT
\times \TT$ as the domain of integration for double integrals to
simplify the notation.

We recall that the Green's function for $-\Delta + \kappa^2$ on $\TT$
can be written as $G(x-y) = -\frac{1}{2\pi} \ln |x-y| + O(|x-y|)$
\cite{m:cmp10}. With the help of this fact, together with
\eqref{pOmell} and \eqref{diamOmi}, for $\eps$ sufficiently small we
have the following estimate for the self-interaction energy:
\begin{align}
  \bar E_\mathrm{self}^\eps & := 2\sum_i \iint
  G(x-y)d\mu_i^{\varepsilon}(x)d\mu_i^{\varepsilon}(y) \nonumber
  \\ %\label{T1.12}
  &\geq -\frac{1}{\pi} \sum_i \iint \( \ln |x-y| + C \)
  d\mu_i^{\varepsilon}(x)d\mu_i^{\varepsilon}(y). \nonumber \\
  & = -\frac{1}{\pi |\ln \varepsilon|^2}\sum_i
  \int_{\overline{\Omega}_i^+} \int_{\overline{\Omega}_i^+} \( \ln
  (\varepsilon^{1/3} |\ln \varepsilon|^{2/3}|\overline{x}-
  \overline{y}|) + C \) d \bar x \, d \bar y, \label{T1.13}
\end{align} 
for some $C > 0$ independent of $\eps$, where in equation
(\ref{T1.13}) we have rescaled coordinates $\bar x = \eps^{-1/3} |\ln
\eps|^{1/3}x$, $\bar y=\eps^{-1/3}|\ln \eps|^{1/3}$ and introduced the
rescaled versions $\overline{\Omega}_i^+$ of $\Omega_i^+$.  Expanding
the logarithm in (\ref{T1.13}) and using \eqref{ZO.2} and
\eqref{diamOmi}, we obtain that $\bar E_\mathrm{self}^\eps$ can be
bounded from below as follows:
\begin{align}
  \bar E_\mathrm{self}^\eps & \geq {1 \over |\ln \eps|} \sum_i
  |A_i^\eps|^2 \left(\frac{1}{3\pi} - C\left(\frac{\ln |\ln
        \varepsilon|}{|\ln \varepsilon|}\right) - \frac{1}{\pi
      |A_i^\eps|^2 |\ln \varepsilon|} \int_{\overline{\Omega}_i^+}
    \int_{\overline{\Omega}_i^+} \ln |\overline x - \overline y| \,
    d\overline x \, d\overline y
  \right) \nonumber \\
  &\label{T1.101} \geq {1 \over |\ln \eps|} \sum_i
  |A_i^\eps|^2\left(\frac{1}{3\pi} - C\left(\frac{\ln |\ln
        \varepsilon|}{|\ln \varepsilon|}\right) - \frac{1}{\pi |\ln
      \varepsilon|}\ln P_i^\eps \right) \\
  & \geq {1 \over |\ln \eps|} \sum_i |A_i^\eps|^2\left(\frac{1}{3\pi}
    - C\left(\frac{\ln |\ln \varepsilon|}{|\ln \varepsilon|}\right)
  \right), \nonumber
\end{align}
for some $C > 0$ independent of $\eps$ (which changes from line to
line).

Now observe that the term in parentheses appearing in the right-hand
side of (\ref{T1.101}) is positive for $\varepsilon$ sufficiently
small. Using this and the fact that $A_i^\eps \geq \tilde A_i^\eps$,
from (\ref{T1.101}) we obtain
\begin{align}
  \label{T1.11} 
  \bar E_\mathrm{self}^\eps \geq {1 \over |\ln \eps|} \sum_i |\tilde
  A_i^\eps|^2\left(\frac{1}{3\pi} - C\left(\frac{\ln |\ln
        \varepsilon|}{|\ln \varepsilon|}\right) \right),
\end{align}
where $C>0$ is a constant independent of $\ep$. It is also clear from the
definition of $\tilde A_i^\eps$ that there exists a constant $c>0$
such that
\begin{align}
  |\tilde A_i^\eps|^2 \leq cA_i^\eps.
\end{align}
Combining this inequality with (\ref{T1.11}) and choosing any $\eta >
0$, for $\eps$ small enough we have $\eta > Cc\frac{\ln |\ln
  \varepsilon|}{|\ln \varepsilon|^2}$ and, therefore, from
\eqref{EGrho} we obtain
\begin{align}\bar E^{\varepsilon}[u^{\varepsilon}] \geq
  &I_\mathrm{def}^{\varepsilon} + \frac{1}{|\ln \varepsilon|} \sum_i
  \left( \sqrt{4 \pi A_i^\eps} - \(\frac{2\bar \delta}{\kappa^2} +
    \eta \) A_i^\eps+ \frac{1}{3\pi} |\tilde A_i^\eps|^2 \right)
  \nonumber \\
  &+ 2 \iint G_\rho(x-y)d\mu^{\varepsilon}(x)
  d\mu^{\varepsilon}(y). \label{T1.4}
\end{align}

\noindent {\em Step 2: Optimization over $A_i^\eps$.}\\

\noindent 
Focusing on the second term in the right-hand side of (\ref{T1.4}), we
define
\begin{align}
  \label{eq:f}
f(x) := \frac{2 \sqrt{\pi}}{\sqrt{x}} + \frac{1}{3\pi} x,  
\end{align}
and observe that $f$ is strictly convex and attains its minimum of
$3^{2/3}$ at $x= 3^{2/3} \pi$, with
\begin{equation}
  \label{fder} f''(x) = \frac{3 \sqrt{\pi}}{2
    x^{5/2}}.
\end{equation} 
We claim that we can bound the second term in the right-hand side of
(\ref{T1.4}) from below by the sum $I + II + III$ of the following
three terms:
\begin{align}
  \label{E33} I & = \frac{1}{|\ln \varepsilon|} \(3^{2/3} - \frac{2
    \bar \delta}{\kappa^2} - \eta \)\sum_i A_i^\eps + \frac{1}{|\ln
    \varepsilon|}\sum_{A_i^\eps > 3^{2/3} \pi \gamma^{-1}}
  3^{2/3}(3^{-1} \gamma^{-1}-1) A_i^\eps, \\
  II & = \frac{1}{|\ln \varepsilon|} \label{E4a}
  \frac{\gamma^{5/2}}{4\pi^2 \cdot 3^{2/3}} \sum_{A_i^\eps < 3^{2/3}
    \pi \gamma} A_i^\eps
  (A_i^\eps -  3^{2/3} \pi )^2, \\
  III & = \frac{1}{|\ln \varepsilon|} \label{E5a}
  \frac{\gamma^{7/2}}{4 \pi} \sum_{3^{2/3} \pi \gamma \leq A_i^\eps
    \leq 3^{2/3} \pi \gamma^{-1} } (A_i^\eps- 3^{2/3} \pi)^2.
\end{align}
Before proving this, observe that defining
\begin{align}
  \label{Meps}
  M^\eps := \bar E^\eps[u^\eps] - \frac{1}{|\ln \varepsilon|}
  \(3^{2/3} - \frac{2 \bar \delta}{\kappa^2} - \eta \)\sum_i A_i^\eps
  - 2 \iint G_\rho(x - y) d \mu^\eps(x) d \mu^\eps(y),
\end{align}
we have from \eqref{T1.4} and \eqref{E33}--\eqref{E5a} that if
$I_\gamma^\eps$ is as in Theorem \ref{almost}, then
\begin{align}\label{Mlowerbound}
  M^\eps \geq \frac{c_1 }{|\ln \eps|} \sum_{i \notin
    I_{\gamma}^{\eps}} A_i^{\eps} + \frac{c_2}{|\ln \eps|} \sum_{i \in
    I_{\gamma}^{\eps}} (A_i^{\eps}-3^{2/3}\pi)^2 +
  I_\mathrm{def}^{\eps} \geq 0 \qquad \forall \gamma \in (0,
  \tfrac13),
\end{align}
for some constants $c_1,c_2 > 0$ depending only on $\gamma$.

We now argue in favor of the lower bound based on
\eqref{E33}--\eqref{E5a}. First observe that by \eqref{P1.N34} we have
for all $A_i^\eps \geq 3^{2/3} \pi \gamma^{-1}$:
\begin{align}
  \label{I1.2} 
  \sqrt{4 \pi A_i^\eps} + \frac{1}{3\pi} |\tilde A_i^\eps|^2 -
  \(\frac{2 \bar \delta}{\kappa^2} + \eta \)A_i^\eps \geq \left(
    3^{2/3} - \frac{2\bar \delta}{\kappa^2} - \eta \right) A_i^\eps +
  3^{2/3} (3^{-1} \gamma^{-1}-1) A_i^\eps.
\end{align}
When $A_i^\eps < 3^{2/3} \pi \gamma^{-1}$, which corresponds to both
\eqref{E4a} and \eqref{E5a}, we use the convexity of $f$ and
(\ref{fder}):
\begin{align}
  \sqrt{4 \pi A_i^\eps} & + \frac{1}{3\pi} |\tilde A_i^\eps|^2 -
  \(\frac{2 \bar \delta}{\kappa^2} + \eta \)A_i^\eps = A_i^\eps\(
  {2\sqrt{\pi} \over \sqrt{A_i^\eps}} +\frac{1}{3\pi}A_i^\eps -
  \frac{2\bar \delta}{\kappa^2} - \eta \) \nonumber \\
  & = A_i^\eps \left( f(A_i^\eps) - { 2 \bar \delta \over \kappa^2} -
    \eta \right) \label{I1.1}
  \\
  & \geq \(3^{2/3} - \frac{2\bar \delta}{\kappa^2} - \eta \)A_i^\eps +
  \frac12 A_i^\eps f'' \left( 3^{2/3} \pi \gamma^{-1} \right)
  (A_i^\eps - 3^{2/3} \pi)^2, \nonumber
\end{align}
where the last line follows from the second order Taylor formula for
$f(x)$ about $x= 3^{2/3}\pi$ and the fact that $f''(x)$ is decreasing.
Combining \eqref{Meps}, (\ref{I1.2}) and (\ref{I1.1}) yields $M^\eps
\geq I + II + III$.

Now using \eqref{Meps} and \eqref{Mlowerbound} with $\gamma$
sufficiently small, we deduce that
\begin{align}
  \bar E^{\varepsilon}[u^{\varepsilon}] \geq \frac{1}{|\ln
    \varepsilon|} \(3^{2/3} - \frac{2 \bar \delta}{\kappa^2} - \eta \)
  \sum_i A_i^\eps + 2 \iint G_\rho(x-y)d\mu^{\varepsilon}(x)
  d\mu^{\varepsilon}(y).\label{beforelimit}
\end{align}
\bigskip

\noindent {\em Step 3: Passage to the limit.}

\bigskip

\noindent 
%We claim that we may now
%conclude that
%\begin{equation}
%  \limsup_{\varepsilon \to 0} \frac{1}{|\ln \varepsilon|} \sum_{i}
%  A_i^\eps < + \infty.\label{S2.1}
%\end{equation}
%Indeed from the a priori bound (\ref{ZO.1}) and (\ref{E3}) and this
%choice of $\delta$ we have
%\begin{equation}\label{S2.2} \limsup_{\varepsilon \to 0}
%  \frac{3^{2/3}}{2} \frac{1}{|\ln \varepsilon|}
%  \sum_{A_i^\eps > \frac{12 \bar \delta \pi}{\kappa^2}}
%  A_i^\eps < +\infty.\end{equation}
%From the isoperimetric inequality and the a priori bound (\ref{ZO.2})
%we have
%\begin{equation}\label{S2.3}
%  \limsup_{\varepsilon \to 0} \frac{1}{|\ln
%    \varepsilon|}\sum_{A_i^\eps \leq \frac{12 \bar \delta
%      \pi}{\kappa^2}} A_i^\eps \leq \limsup_{\varepsilon
%    \to 0} \frac{C}{|\ln \varepsilon|} \sum_{A_i^\eps
%    \leq \frac{12 \bar \delta \pi}{\kappa^2}}
%  |A_i^\eps|^{1/2} \stackrel{(\ref{ZO.2})}{<} +\infty.
%\end{equation}
%Combining (\ref{S2.2}) and (\ref{S2.3}) yields (\ref{S2.1}). 
We may now conclude from \eqref{EEE}--(\ref{ZO.1}), \eqref{vdef},
\eqref{vmu}, \eqref{Gmumu} and (\ref{ZO.2}) that
\begin{equation}
  \limsup_{\varepsilon \to 0} \int_{\mathbb{T}_l^2} (|\nabla v^{\varepsilon}|^2 +
  \kappa^2 |v^{\varepsilon}|^2 ) dx < +\infty,
\end{equation}
while $(\mu^\eps)$ are bounded in the sense of measures from
\eqref{mulinf}.  Consequently, up to a subsequence
\begin{align}
  v^{\varepsilon} &\rightharpoonup v \textrm{ in }
  H^1(\TT), \\
  \mu^{\varepsilon} &\stackrel{*}{\rightharpoonup} \mu \textrm{ in }
  C(\TT),
\end{align}
where
\begin{align}
  -\Delta v + \kappa^2 v = \mu
\end{align}
holds in the distributional sense.  Now passing to the limit in
\eqref{beforelimit} and recalling \eqref{muA}, we obtain
\begin{align} 
  \liminf_{\varepsilon \to 0} \bar E^{\varepsilon}[u^{\varepsilon}]
  \geq \left( 3^{2/3} - \frac{2\bar \delta}{\kappa^2} - \eta \right)
  \int d\mu + 2 \iint G_\rho(x-y)d\mu(x)d\mu(y),
\end{align} 
using continuity of $G_\rho$. On the other hand, we have $G_\rho(x -
y) \to G(x - y)$ monotonically from below for each $x \not = y$ as
$\rho \to 0$. Moreover, since $\mu$ satisfies \eqref{dmubdG}, the set
$\{(x, y) \in \TT \times \TT : x = y\}$ is $\mu \otimes
\mu$-negligible.  An application of monotone convergence theorem then
yields
\begin{align} \liminf_{\varepsilon \to 0} \bar
  E^{\varepsilon}[u^{\varepsilon}] \geq \left( 3^{2/3} - \frac{2\bar
      \delta}{\kappa^2} \right) \int d\mu + 2\iint
  G(x-y)d\mu(x)d\mu(y),
\end{align} 
upon sending $\rho \to 0$ and then $\eta \to 0$. \hfill $\Box$

\bigskip

We now argue in favor of the corresponding upper bound in Theorem
\ref{main}. The construction resembles quite closely that of the
vortex construction in \cite{sandier00} for the two dimensional
Ginzburg-Landau functional and indeed we borrow several ideas from
that proof and occasionally
refer the reader to that paper for details.

\subsection{Proof of the Upper Bound, Theorem \ref{main} ii)}
\label{s:rec}

As in the proof of the lower bound, we set $d \mu_i^{\varepsilon}(x)$
as in \eqref{muieps}, so that $\mu^{\varepsilon} = \sum
\mu_i^\eps$. If $\int_\TT d \mu = 0$, there is nothing to
prove. Otherwise, using a mollification with a strictly positive
mollifier we can always approximate the measure $\mu$ by a measure
with a smooth strictly positive density and retrieve a recovery
sequence by a standard diagonal argument. Hence without loss of
generality in this section we assume that
% the approximation argument of Proposition II.2 in
% \cite{sandier00} we may assume that the limit measure $\mu$ satisfies
\begin{align}\label{approx}
  d\mu(x) = g(x) dx, \qquad c \leq g \leq C,
\end{align}
for some $C > c > 0$.  

\bigskip

\noindent {\em Step 1: Construction of the configuration.}

\medskip

\noindent We claim that for $\eps$ sufficiently small it is possible
to place a total of $N(\eps)$ disjoint spherical droplets, where
\begin{align}
  \label{eq:Neps}
  N(\varepsilon) = \frac{1}{3^{2/3}} \frac{|\ln \varepsilon|}{\pi}
  \mu(\TT) + o(|\ln \eps|),
\end{align}
with centers $\{a_i\}$ in $\TT$ and radius
\begin{align}
  \label{eq:r}
  r = 3^{1/3}\varepsilon^{1/3} |\ln \varepsilon|^{-1/3},
\end{align}
and satisfying for all $i \neq j$
\begin{equation}\label{dist}
  d(\eps) := \min |a_i - a_j| \geq \frac{C}{\sqrt{N(\varepsilon)}},
\end{equation}
for some constant $C > 0$ depending only on $\mu$. Indeed, given $\mu$
satisfying \eqref{approx}, for $\eps$ sufficiently small we can
partition $\TT$ into disjoint squares $\{K_i\}$ of side length
$\eta_\eps > 0$ (hereafter simply denoted $\eta$) satisfying
\begin{align}\label{ep13}
  |\ln \varepsilon|^{-1/2} \ll \eta \ll 1.
\end{align} 
In each $K_i$ we place
\begin{align}
  \label{NKi}
  N_{K_i}(\varepsilon) = \bigg\lfloor \frac{1}{3^{2/3}} \frac{|\ln
    \varepsilon|}{\pi} \mu(K_i) \bigg\rfloor
\end{align}
points $a_i$ (here $m = \lfloor x \rfloor$ denotes the smallest
integer $m \leq x$) satisfying $(\ref{dist})$ and in addition
\begin{equation}
  \dist (a_i,\partial K_i) \geq
  \frac{C}{\sqrt{N(\varepsilon)}}, \qquad N(\eps) := \sum_i N_{K_i}.
\end{equation}
As argued in \cite{sandier00}, our ability to do this follows from the
estimate:
\begin{equation}
  c \eta^2 \leq \mu(K_i) \leq C \eta^2,
\end{equation}
which follows from $(\ref{approx})$ together with \eqref{ep13}. We
finally define our configuration $u^\eps$ by setting the connected
components $\Omega_i^+$ of $\Omega^+ = \{ u^\eps = +1\}$ to be balls
of radius $r$ from \eqref{eq:r} centered at $a_i$, i.e. $\Omega_i^+ :=
B_r(a_i)$. We set $u^\eps=-1$ in the complement of these balls.

With these choices we have
\begin{align}
  \bar E^\eps[u^\eps] &= \frac{2\pi \cdot 3^{1/3} N(\eps) } {|\ln
    \eps|} - \frac{2 \pi \cdot 3^{2/3} N(\eps) \bar \delta }{|\ln
    \eps| \kappa^2} + 2 \iint
  G(x-y)d\mu^{\varepsilon}(x)d\mu^{\varepsilon}(y)  \nonumber \\
  & = \frac{2}{3^{1/3}} \mu(\TT)- \frac{2\bar \delta}{\kappa^2}
  \mu(\TT) + 2 \iint G(x-y)d\mu^{\varepsilon}(x)d\mu^{\varepsilon}(y)
  + o(1). \label{T2.1a}
 \end{align}
 The main point of the rest of the proof is to show that the integral
 term in $(\ref{T2.1a})$ converges to $\iint G(x - y) d\mu(x) d\mu(y)
 + 3^{-1/3} \int d \mu$, with the non-trivial last term coming from
 the self-interaction of the droplets. To prove that these are the
 only contributions to the limit energy, we need to use the fact that
 the droplets do not concentrate too much as $\eps \to 0$.

\medskip

\noindent {\em Step 2: Convergence of the configurations.}

\bigskip
\noindent
Defining $\mu^{\varepsilon}$ as before, it is clear from the
construction that
\begin{align}\label{cvmuu}
  \mu^{\varepsilon} &\rightharpoonup \mu \textrm{ in }
  (C(\TT))^*.
\end{align}
Fix $\rho > 0$ sufficiently small (depending only on $\kappa$ and
$\ell$) and consider $G_\rho(x - y)$ defined as in \eqref{Geta}. By
the continuity of $G_\rho$ in $\TT$ we have
\begin{equation}\label{T1.22} \lim_{\varepsilon \to 0}
  \iint
  G_\rho(x-y)d\mu^{\varepsilon}(x)d\mu^{\varepsilon}(y)
  = \iint
  G_\rho(x-y)d\mu(x)d\mu(y).\end{equation}
Now, let $I_{\rho}$ be the collection of indices $(i,j)$ such that
$0 < |a_i - a_j| < \rho$. Then for $\eps$ small enough we can write
\begin{align}\label{T2.22}
  \iint& \big( G(x - y) - G_\rho(x - y) \big) d\mu^{\varepsilon}(x)
  d\mu^{\varepsilon}(y) \nonumber \\ & \leq
  \sum_{i=1}^{N(\varepsilon)} \iint G(x - y) d\mu_i^{\varepsilon}(x)
  d\mu_i^{\varepsilon}(y) + \sum_{(i,j) \in I_{\rho}} \iint G(x - y)
  d\mu_i^{\varepsilon}(x) d\mu_j^{\varepsilon}(y) \\ &\leq
  \frac{1}{6\pi |\ln \varepsilon| } \sum_{i=1}^{N(\varepsilon)}
  |A_i^\eps|^2 + \frac{C \ln|\ln \varepsilon|}{|\ln \varepsilon|} +
  {C' \over |\ln \eps|^2} \sum_{(i,j) \in I_{\rho}} A_i^\eps A_j^\eps
  \left| \ln \dist (\Omega_i^+, \Omega_j^+) \right|, \nonumber
\end{align}
for some $C, C' > 0$ independent of $\eps$ or $\rho$, where $A_i^\eps
= 3^{2/3} \pi $ and we expanded the Green's function as in
\eqref{T1.13} in the proof of the lower bound.  Now, for $k = 1, 2,
\ldots, K_\rho(\eps)$, with $K_\rho(\eps) := \lfloor \rho / d(\eps)
\rfloor$, let $I_\rho^k \subset I_\rho$ be disjoint sets consisting of
all indices $(i, j)$ such that $k d(\eps) \leq |a_i - a_j| < (k + 1)
d(\eps)$.  Since by the result on optimal packing density of disks in
the plane \cite{toth40} we have $|I_\rho^k| \leq c k N(\eps)$ for some
universal $c > 0$ (here again $|I_\rho^k|$ denotes the cardinality of
$I_\rho^k$), in view of \eqref{eq:Neps} it holds that
\begin{align}
  \label{eq:shortrange}
  \frac{1}{|\ln \eps|^2} \sum_{(i,j) \in I_{\rho}} A_i^\eps A_j^\eps
  \left| \ln \dist (\Omega_i^+, \Omega_j^+) \right| \leq {C N(\eps)
    \over |\ln \eps|^2} \sum_{k=1}^{K_\rho(\eps)} k |\ln (k d(\eps))|
  \nonumber \\ \leq {2 C N(\eps) \over |\ln \eps|^2 d^2(\eps) }
  \int_{d(\eps)}^{\rho} t |\ln t| dt \leq C' \left( { |\ln d(\eps)|
      \over |\ln \eps|} + \rho^2 |\ln \rho| \right) \leq 2 C' \rho^2
  |\ln \rho|,
\end{align}
for some $C, C' > 0$ independent of $\eps$ or $\rho$, when $\eps$ and
$\rho$ are sufficiently small.  Therefore, from \eqref{eq:r} and
(\ref{T2.22}) we obtain
\begin{align}\label{T2.21}
  \limsup_{\eps \to 0} \iint \big( G(x - y) - G_\rho(x - y) \big)
  d\mu^{\varepsilon}(x) d\mu^{\varepsilon}(y) \leq 2^{-1} \cdot
  3^{-1/3} + o(\rho).
\end{align}
Finally combining (\ref{T2.21}) with (\ref{T2.1a}) and (\ref{T1.22}),
upon sending $\eps \to 0$, then $\rho \to 0$ and applying the monotone
convergence theorem we have
\begin{align} \lim_{\varepsilon \to 0} \bar
  E^{\varepsilon}[u^{\varepsilon}] \leq \left( 3^{2/3} - \frac{2\bar
      \delta}{\kappa^2} \right) \int d\mu + 2\iint
  G(x-y)d\mu(x)d\mu(y),
\end{align} 
as required. The fact that $v^{\varepsilon} \rightharpoonup v$ follows
from \eqref{cvmuu} and the uniform bounds just demonstrated on the
terms involving the Green's function in (\ref{T2.1a}), from which it
follows that (\ref{PDE}) is satisfied distributionally. \hfill $\Box$

\section{Proof of Theorem \ref{almost}}
\label{sec:proof-theorem-almost}

In the proof of Sec. \ref{sec:main}, we have in fact established
Theorem \ref{almost}, which is clear by \eqref{Mlowerbound}. Indeed,
we have for a sequence of almost minimizers $(u^{\varepsilon})$:
\begin{align}
  \lim_{\varepsilon \to 0} E^{\varepsilon}[u^{\varepsilon}] - E_0[\mu]
  = 0.
\end{align}
Observing that $M^\eps$ defined in \eqref{Meps} does not contribute to
$E_0[\mu]$, we have established that $M^\eps \to 0$ as $\eps \to 0$
for any $\gamma < \frac13$ and, as a consequence, we obtain
(\ref{cor1.1})--(\ref{cor1.2}) for, say, $\gamma = \frac16$. Then it
is easy to see from the definition of $I_\gamma^\ep$ that the
statement of the Theorem, in fact, holds for any $\gamma \in (0, 1)$.
\hfill $\Box$

\section{Proof of Theorem \ref{main2}}
\label{sec:proof-theor-refm}

We now turn to the proof of Theorem \ref{main2} extending the result
of Theorem \ref{main} for the sharp interface energy $E^\eps$ to the
diffuse interface energy $\mathcal E^\eps$. The proof proceeds by a
refinement of the ideas of \cite[Sec. 4]{m:cmp10} to establish
matching upper and lower bounds for $\mathcal E^\eps$ in terms of
$E^\eps$ for sequences with bounded energy. 

\bigskip

\noindent {\em Step 1: Approximate lower bound.}

\bigskip
\noindent In the following, it is convenient to rewrite the energy
\eqref{EE2} in an equivalent form
\begin{align}
  \label{EE3}
  \mathcal E^\eps[u^\eps] = \int_\TT \left( \frac{\eps^2}2 | \nabla
    u^\eps|^2 + W(u^\eps) + \frac12 |\nabla v^\eps|^2 \right) dx,
  \qquad -\Delta v^\eps = u^\eps - \bar u^\eps, \qquad \int_\TT v^\eps
  dx = 0.
\end{align}
Fix any $\delta \in (0, 1)$ and consider a sequence $(u^\eps) \in
\mathcal A^\eps$ such that $\limsup_{\eps \to 0} \|
u^\eps\|_{L^\infty(\TT)} \leq 1$ and $ \mathcal E^\eps[u^\eps] \leq C
\eps^{4/3} |\ln \eps|^{2/3}$ for some $C > 0$ independent of
$\eps$. Then we claim that
\begin{align}
  \label{limsups}
  \limsup_{\eps \to 0} \|u^\eps\|_{L^\infty(\TT)} = 1, \qquad
  \lim_{\eps \to 0} \|v^\eps\|_{L^\infty(\TT)} = 0.
\end{align}
Indeed, for the first statement we have from the definition of
$\mathcal E^\eps$ in \eqref{EE2} that
\begin{align}
  \label{interfaces}
  | \Omega_0^\delta | \leq C \eps^{4/3} |\ln \eps|^{2/3} \delta^{-2},
  \qquad \Omega_0^\delta := \{ -1 + \delta \leq u^\eps \leq 1 - \delta
  \},
\end{align}
for some $C > 0$ independent of $\eps$. Hence, in particular,
$\limsup_{\eps \to 0} \|u^\eps\|_{L^\infty(\TT)} \geq 1$, proving the
first statement of \eqref{limsups}.  To prove the second statement in
\eqref{limsups}, we note that by standard elliptic theory (see, e.g.,
\cite[Theorem 9.9]{gilbarg}) we have $\|v^\eps\|_{W^{2,p}(\TT)} \leq
C'$ for any $p > 2$ and some $C' > 0$ independent of $\eps$ and,
hence, by Sobolev embedding $\|\nabla v^\eps\|_{L^\infty(\TT)} \leq
C''$ for some $C'' > 0$ independent of $\eps$ as well. Therefore,
applying Poincar\'e's inequality, we obtain
\begin{align}
  \label{estveps}
  C \eps^{4/3} |\ln \eps|^{2/3} \geq \mathcal E^\eps[u^\eps] \geq C'
  \int_\TT |v^\eps|^2 dx \geq C'' \|v^\eps\|_{L^\infty(\TT)}^4,
\end{align}
for some $C', C'' > 0$  independent of $\eps$, yielding the claim. 

In view of \eqref{limsups}, for small enough $\eps$ we have
$\|u^\eps\|_{L^\infty(\TT)} \leq 1 + \delta^3$ and
$\|v^\eps\|_{L^\infty(\TT)} \leq \delta^3$, and by the assumption on
energy we may further assume that $\mathcal E^\eps[u^\eps] \leq
\delta^{12}$. Therefore, by \cite[Proposition 4.2]{m:cmp10} there
exists a function $\tilde u_0^\eps \in \mathcal A$ such that
\begin{align}
  \label{lowerbd}
  \mathcal E^\eps[u^\eps] \geq (1 - \delta^{1/2}) E^\eps[\tilde
  u_0^\eps].
\end{align}
In particular, $(\tilde u_0^\eps)$ satisfy the assumptions of Theorem
\ref{main}, and, therefore, upon extraction of subsequences we have
$\tilde \mu_0^\eps \rightharpoonup \mu \in \mathcal M^+(\TT) \cap
H^{-1}(\TT)$ in $(C(\TT))^*$, where
\begin{align}
  \label{tildemueps}
  d \tilde \mu_0^\eps(x) := \tfrac12 \eps^{-2/3} |\ln \eps|^{-1/3} (1
  + \tilde u_0^\eps(x)) dx.
\end{align}
Furthermore, recalling that by construction the jump set of $\tilde
u^\eps_0$ is either contained in $\Omega_0^\delta$ or empty, see the
proof of \cite[Lemma 4.1]{m:cmp10}, from \eqref{interfaces} we have
\begin{align}
  \label{diffu0utilde}
  \| \tilde u_0^\eps - u_0^\eps\|_{L^1(\TT)} \leq C \eps^{4/3} |\ln
  \eps|^{2/3} \delta^{-2},
\end{align}
where $u^\eps_0$ is given by \eqref{ueps0}, for some $C > 0$
independent of $\eps$. Comparing \eqref{diffu0utilde} with
\eqref{tildemueps}, we then see that $\mu_0^\eps \rightharpoonup \mu$
in $(C(\TT))^*$ as well. The result of part (i) of Theorem \ref{main2}
then follows by the arbitrariness of $\delta > 0$ via a diagonal
process. \qed

\bigskip

\noindent {\em Step 2: Approximate upper bound.}

\bigskip

\noindent First note that if $\mu = 0$, we can choose $u^\eps = \bar
u^\eps$. Indeed, we have $\eps^{-4/3} |\ln \eps|^{-2/3} \mathcal
E^\eps[\bar u^\eps] = \ell^2 W(\bar u^\eps) = {\ell^2 \bar \delta^2
  \over 2 \kappa^2} + o(1)$ and $\bar u^\eps \to -1$. On the other
hand, if $\int_\TT d \mu > 0$, we can construct the approximate upper
bounds for a suitable lifting of the recovery sequences in the proof
of Theorem \ref{main}(ii) to $\mathcal A^\eps$. Let $(\tilde u_0^\eps)
\in \mathcal A$ be a recovery sequence constructed in
Sec. \ref{s:rec}. This sequence consists of circular droplets of the
optimal radius $r = 3^{1/3} \eps^{1/3} |\ln \eps|^{-1/3} \gg
\eps^{1/2}$ and mutual distance $d \geq C |\ln \eps|^{-1/2} \gg
\eps^{1/2}$, for some $C > 0$ independent of $\eps$. In addition,
since
\begin{align}
  \label{E3}
  E^\eps[\tilde u^\eps_0] = \frac{\eps}{2} \int_\TT |\nabla \tilde
  u_0^\eps| \, dx + 2 \int_\TT \Big( |\nabla \tilde v^\eps|^2 +
  \kappa^2 |\tilde v^\eps|^2 \Big) dx \leq C \eps^{4/3} |\ln
  \eps|^{2/3},
\end{align}
where $\tilde v^\eps(x) = \int_\TT G(x - y) (\tilde u_0^\eps(y) - \bar
u^\eps) dy$, for some $C > 0$ independent of $\eps$, by the argument
of \eqref{estveps} one can see that $\lim_{\eps \to 0} \|\tilde
v^\eps\|_{L^\infty(\TT)} = 0$. Therefore, for any $\delta \in (0, 1)$
and $\eps > 0$ sufficiently small we have $\|\tilde
v^\eps\|_{L^\infty(\TT)} \leq \delta$ and $E^\eps[\tilde u^\eps_0]
\leq \delta^{5/2}$. We can then apply \cite[Proposition 4.3]{m:cmp10}
to obtain a function $u^\eps \in \mathcal A^\eps$ such that
\begin{align}
  \label{upperbd}
  \mathcal E^\eps[u^\eps] \leq (1 + \delta^{1/2}) E^\eps[\tilde
  u_0^\eps].
\end{align}
Furthermore, by the construction of $u^\eps$ (see
\cite[Eqs. (4.31)--(4.33)]{m:cmp10}) and arbitrariness of $\delta >
0$, we also have $\limsup_{\eps \to 0} \|u^\eps\|_{L^\infty(\TT)} =
1$, and
\begin{align}
  \label{diffu0u}
  \| \tilde u_0^\eps - u_0^\eps\|_{L^1(\TT)} \leq C \eps^{4/3} |\ln
  \eps|^{2/3},
\end{align}
for some $C > 0$ independent of $\eps$, where $u^\eps_0$ is given by
\eqref{ueps0}, and we used \eqref{E3}. Hence $\mu_0^\eps
\rightharpoonup \mu = \lim_{\eps \to 0} \tilde \mu^\eps_0$ in
$(C(\TT))^*$.  The result of part (ii) of Theorem \ref{main2} again
follows by arbitrariness of $\delta > 0$ via a diagonal process. \qed

\begin{remark}
  It is possible to chose $\delta = \eps^{\alpha}$ for $\alpha > 0$
  sufficiently small in the arguments of the proof of Theorem
  \ref{main2}. Therefore, given a sequence of minimizers $(u^\eps) \in
  \mathcal A^\eps$ of $\mathcal E^\eps$ and the corresponding sequence
  $(u_0^\eps) \in \mathcal A$ of minimizers of $E^\eps$, one has
  \begin{align}
    \label{liminfEequiv}
    \eps^{-4/3} |\ln \eps|^{-2/3} \mathcal E^\eps[u^\eps] =
    \eps^{-4/3} |\ln \eps|^{-2/3} E^\eps[u_0^\eps] + O(\eps^\alpha),
  \end{align}
  for some $\alpha \ll 1$, as $\eps \to 0$.
\end{remark}

\section{Proof of Theorem \ref{almost2}}
\label{sec:proof-theor-refalm}

Let $(u^\eps)$ be a sequence from Theorem \ref{main2}(ii). Arguing as
in Step 1 of the proof of Theorem \ref{main2}, for every $\delta > 0$
sufficiently small there exists a sequence $(\tilde u^\eps_0) \in
\mathcal A$ such that \eqref{lowerbd} holds, the jump set of $\tilde
u^\eps_0$ is contained in $\{ -1 + \delta \leq u^\eps \leq 1 -
\delta\}$, and if $\tilde \mu_0^\eps$ is defined via
\eqref{tildemueps}, then $\tilde \mu_0^\eps \rightharpoonup \mu$ in
$(C(\TT))^*$. On the other hand, applying the result of Theorem
\ref{main}(i), we obtain
\begin{multline}
  \label{EalmE}
  E^0[\mu] \geq \limsup_{\eps \to 0} \eps^{-4/3} |\ln \eps|^{-2/3}
  \mathcal E^\eps[u^\eps] \geq (1 - \delta^{1/2}) \limsup_{\eps \to 0}
  \eps^{-4/3} |\ln \eps|^{-2/3} E^\eps[\tilde u_0^\eps] \\ \geq (1 -
  \delta^{1/2}) \liminf_{\eps \to 0} \eps^{-4/3} |\ln \eps|^{-2/3}
  E^\eps[\tilde u_0^\eps] \geq (1 - \delta^{1/2}) E^0[\mu].
\end{multline}
Therefore, in view of arbitrariness of $\delta > 0$ we conclude that
$(\tilde u_0^\eps)$ is a sequence of almost minimizers of $E^\eps$
with prescribed density $\mu$ by a diagonal process. As a consequence,
Theorem \ref{almost} applies to $(\tilde u^\eps_0)$. Moreover, by
\eqref{diffu0utilde} and the fact that 
\begin{align}
  \label{supzu}
  \lim_{\eps \to 0} \eps^{-2/3} |\ln \eps|^{-1/3} |\{ u^\eps > 0 \}| =
  \int_\TT d \mu > 0,
\end{align}
we obtain
\eqref{limOmpu}. \qed

\bigskip

\noindent \textbf{Acknowledgments} The research of D. G.  was
partially supported by an NSERC PGS D award. The work of C. B. M. was
supported, in part, by NSF via grants DMS-0718027 and DMS-0908279. The
research of S. S. was supported by the EURYI award. C. B. M. would
like to acknowledge valuable discussions with H. Kn\"upfer and
M. Novaga.

\bibliography{../nonlin,../mura,../stat}

\begin{thebibliography}{10}

\bibitem{alberti09}
G.~Alberti, R.~Choksi, and F.~Otto.
\newblock Uniform energy distribution for an isoperimetric problem with
  long-range interactions.
\newblock {\em J. Amer. Math. Soc.}, 22:569--605, 2009.

\bibitem{ambrosio01}
L.~Ambrosio, V.~Caselles, S.~Masnou, and J.-M. Morel.
\newblock Connected components of sets of finite perimeter and applications to
  image processing.
\newblock {\em J. Eur. Math. Soc.}, 3:39--92, 2001.

\bibitem{bates99}
F.~S. Bates and G.~H. Fredrickson.
\newblock Block copolymers -- designer soft materials.
\newblock {\em Physics Today}, 52:32--38, 1999.

\bibitem{braides}
A.~Braides.
\newblock {\em {$\Gamma$}-convergence for beginners}, volume~22 of {\em Oxford
  Lecture Series in Mathematics and its Applications}.
\newblock Oxford University Press, Oxford, 2002.

\bibitem{braides08}
A.~Braides and L.~Truskinovsky.
\newblock Asymptotic expansions by {$\Gamma$}-convergence.
\newblock {\em Continuum Mech. Thermodyn.}, 20:21--62, 2008.

\bibitem{chen93}
L.~Q. Chen and A.~G. Khachaturyan.
\newblock Dynamics of simultaneous ordering and phase separation and effect of
  long-range {Coulomb} interactions.
\newblock {\em Phys. Rev. Lett.}, 70:1477--1480, 1993.

\bibitem{choksi01}
R.~Choksi.
\newblock Scaling laws in microphase separation of diblock copolymers.
\newblock {\em J. Nonlinear Sci.}, 11:223--236, 2001.

\bibitem{choksi08}
R.~Choksi, S.~Conti, R.~V. Kohn, and F.~Otto.
\newblock Ground state energy scaling laws during the onset and destruction of
  the intermediate state in a {Type-I} superconductor.
\newblock {\em Comm. Pure Appl. Math.}, 61:595--626, 2008.

\bibitem{choksi98}
R.~Choksi and R.~V. Kohn.
\newblock Bounds on the micromagnetic energy of a uniaxial ferromagnet.
\newblock {\em Comm. Pure Appl. Math.}, 51:259--289, 1998.

\bibitem{choksi99}
R.~Choksi, R.~V. Kohn, and F.~Otto.
\newblock Domain branching in uniaxial ferromagnets: a scaling law for the
  minimum energy.
\newblock {\em Commun. Math. Phys.}, 201:61--79, 1999.

\bibitem{choksi04}
R.~Choksi, R.~V. Kohn, and F.~Otto.
\newblock Energy minimization and flux domain structure in the intermediate
  state of a {Type-I} superconductor.
\newblock {\em J. Nonlinear Sci.}, 14:119--171, 2004.

\bibitem{choksi11siads}
R.~Choksi, M.~Maras, and J.~F. Williams.
\newblock {2D} phase diagram for minimizers of a {Cahn--Hilliard} functional
  with long-range interactions.
\newblock {\em SIAM J. Appl. Dyn. Syst.}, 10:1344--1362, 2011.

\bibitem{choksi10}
R.~Choksi and M.~A. Peletier.
\newblock Small volume fraction limit of the diblock copolymer problem: {I.
  Sharp} interface functional.
\newblock {\em SIAM J. Math. Anal.}, 42:1334--1370, 2010.

\bibitem{choksi11}
R.~Choksi and M.~A. Peletier.
\newblock Small volume fraction limit of the diblock copolymer problem: {II.
  Diffuse} interface functional.
\newblock {\em SIAM J. Math. Anal.}, 43:739--763, 2011.

\bibitem{degennes79}
P.~G. {de Gennes}.
\newblock Effect of cross-links on a mixture of polymers.
\newblock {\em J. de Physique -- Lett.}, 40:69--72, 1979.

\bibitem{desimone00}
A.~DeSimone, R.~V. Kohn, S.~M\"uller, and F.~Otto.
\newblock Magnetic microstructures---a paradigm of multiscale problems.
\newblock In {\em ICIAM 99 (Edinburgh)}, pages 175--190. Oxford Univ. Press,
  2000.

\bibitem{emery93}
V.~J. Emery and S.~A. Kivelson.
\newblock Frustrated electronic phase-separation and high-temperature
  superconductors.
\newblock {\em Physica C}, 209:597--621, 1993.

\bibitem{toth40}
L.~Fejes~T\'oth.
\newblock \"{U}ber einen geometrischen {S}atz.
\newblock {\em Math. Z.}, 46:83--85, 1940.

\bibitem{friesecke06}
G.~Friesecke, R.~D. James, and S.~M{\"u}ller.
\newblock A hierarchy of plate models derived from nonlinear elasticity by
  {Gamma}-convergence.
\newblock {\em Arch. Ration. Mech. Anal.}, 180:183--236, 2006.

\bibitem{fusco08}
N.~Fusco, F.~Maggi, and A.~Pratelli.
\newblock The sharp quantitative isoperimetric inequality.
\newblock {\em Ann. of Math.}, 168:941--980, 2008.

\bibitem{gilbarg}
D.~Gilbarg and N.~S. Trudinger.
\newblock {\em Elliptic Partial Differential Equations of Second Order}.
\newblock Springer-Verlag, Berlin, 1983.

\bibitem{glotzer95}
S.~Glotzer, E.~A. Di~Marzio, and M.~Muthukumar.
\newblock Reaction-controlled morphology of phase-separating mixtures.
\newblock {\em Phys. Rev. Lett.}, 74:2034--2037, 1995.

\bibitem{gms11b}
D.~Goldman, C.~B. Muratov, and S.~Serfaty.
\newblock The {$\Gamma$}-limit of the two-dimensional {Ohta-Kawasaki} energy.
  {II. Droplet} arrangement at the sharp interface level via the renormalized
  energy.
\newblock {\em (submitted to Arch. Rational Mech. Anal.)}, 2012.
\newblock Preprint. arXiv:1210.5098.

\bibitem{grimes79}
C.~C. Grimes and G.~Adams.
\newblock Evidence for a liquid-to-crystal phase transition in a classical,
  two-dimensional sheet of electrons.
\newblock {\em Phys. Rev. Lett.}, 42:795--798, 1979.

\bibitem{hubert}
A.~Hubert and R.~Sch\"afer.
\newblock {\em Magnetic Domains}.
\newblock Springer, Berlin, 1998.

\bibitem{huebener}
R.~P. Huebener.
\newblock {\em Magnetic flux structures in superconductors}.
\newblock Springer-Verlag, Berlin, 1979.

\bibitem{km:jns11}
H.~Kn\"upfer and C.~B. Muratov.
\newblock Domain structure of bulk ferromagnetic crystals in applied fields
  near saturation.
\newblock {\em J. Nonlinear Sci.}, 21:921--962, 2011.

\bibitem{km:cpam13}
H.~Kn\"upfer and C.~B. Muratov.
\newblock On an isoperimetric problem with a competing non-local term. {II.
  The} general case.
\newblock {\em Commun. Pure Appl. Math.}, 2013 (to appear).

\bibitem{kohn07iciam}
R.~V. Kohn.
\newblock Energy-driven pattern formation.
\newblock In {\em International {C}ongress of {M}athematicians. {V}ol. {I}},
  pages 359--383. Eur. Math. Soc., Z\"urich, 2007.

\bibitem{lebris05}
C.~Le~Bris and P.-L. Lions.
\newblock From atoms to crystals: a mathematical journey.
\newblock {\em Bull. Amer. Math. Soc. (N.S.)}, 42:291--363, 2005.

\bibitem{lieb81}
E.~H. Lieb.
\newblock Thomas-{F}ermi and related theories of atoms and molecules.
\newblock {\em Rev. Mod. Phys.}, 53:603--641, 1981.

\bibitem{lundqvist}
S.~Lundqvist and N.~H. March, editors.
\newblock {\em Theory of inhomogeneous electron gas}.
\newblock Plenum Press, New York, 1983.

\bibitem{modica87}
L.~Modica.
\newblock The gradient theory of phase transitions and the minimal interface
  criterion.
\newblock {\em Arch. Rational Mech. Anal.}, 98:123--142, 1987.

\bibitem{modica77}
L.~Modica and S.~Mortola.
\newblock Un esempio di {$\Gamma$}-convergenza.
\newblock {\em Boll. Un. Mat. Ital. B}, 14:285--299, 1977.

\bibitem{muller93}
S.~M\"uller.
\newblock Singular perturbations as a selection criterion for periodic
  minimizing sequences.
\newblock {\em Calc. Var. PDE}, 1:169--204, 1993.

\bibitem{m:phd}
C.~B. Muratov.
\newblock {\em Theory of domain patterns in systems with long-range
  interactions of Coulombic type}.
\newblock Ph. D. Thesis, Boston University, 1998.

\bibitem{m:pre02}
C.~B. Muratov.
\newblock Theory of domain patterns in systems with long-range interactions of
  {Coulomb} type.
\newblock {\em Phys. Rev. E}, 66:066108 pp. 1--25, 2002.

\bibitem{m:cmp10}
C.~B. Muratov.
\newblock Droplet phases in non-local {Ginzburg-Landau} models with {Coulomb}
  repulsion in two dimensions.
\newblock {\em Comm. Math. Phys.}, 299:45--87, 2010.

\bibitem{muthukumar97}
M.~Muthukumar, C.~K. Ober, and E.~L. Thomas.
\newblock Competing interactions and levels of ordering in self-organizing
  polymeric materials.
\newblock {\em Science}, 277:1225--1232, 1997.

\bibitem{nagaev95}
E.~L. Nagaev.
\newblock Phase separation in high-temperature superconductors and related
  magnetic systems.
\newblock {\em Phys. Uspekhi}, 38:497--521, 1995.

\bibitem{nyrkova94}
I.~A. Nyrkova, A.~R. Khokhlov, and M.~Doi.
\newblock Microdomain structures in polyelectrolyte systems: calculation of the
  phase diagrams by direct minimization of the free energy.
\newblock {\em Macromolecules}, 27:4220--4230, 1994.

\bibitem{ohta86}
T.~Ohta and K.~Kawasaki.
\newblock Equilibrium morphologies of block copolymer melts.
\newblock {\em Macromolecules}, 19:2621--2632, 1986.

\bibitem{radin81}
C.~Radin.
\newblock The ground state for soft disks.
\newblock {\em J. Statist. Phys.}, 26:365--373, 1981.

\bibitem{ren00trusk}
X.~Ren and L.~Truskinovsky.
\newblock Finite scale microstructures in nonlocal elasticity.
\newblock {\em J. Elasticity}, 59:319--355, 2000.

\bibitem{ren07rmp}
X.~Ren and J.~Wei.
\newblock Many droplet pattern in the cylindrical phase of diblock copolymer
  morphology.
\newblock {\em Rev. Math. Phys.}, 19:879--921, 2007.

\bibitem{ren00}
X.~F. Ren and J.~C. Wei.
\newblock On the multiplicity of solutions of two nonlocal variational
  problems.
\newblock {\em SIAM J. Math. Anal.}, 31:909--924, 2000.

\bibitem{sandier00}
E.~Sandier and S.~Serfaty.
\newblock A rigorous derivation of a free-boundary problem arising in
  superconductivity.
\newblock {\em Ann. Sci. \'Ecole Norm. Sup. (4)}, 33:561--592, 2000.

\bibitem{sandier}
E.~Sandier and S.~Serfaty.
\newblock {\em Vortices in the magnetic {G}inzburg-{L}andau model}.
\newblock Progress in Nonlinear Differential Equations and their Applications,
  70. Birkh\"auser Boston Inc., Boston, MA, 2007.

\bibitem{sandier12}
E.~Sandier and S.~Serfaty.
\newblock From the {Ginbzurg-Landau} model to vortex lattice problems.
\newblock {\em Comm. Math. Phys.}, 313:635--743, 2012.

\bibitem{seul95}
M.~Seul and D.~Andelman.
\newblock Domain shapes and patterns: the phenomenology of modulated phases.
\newblock {\em Science}, 267:476--483, 1995.

\bibitem{spadaro09}
E.~Spadaro.
\newblock Uniform energy and density distribution: diblock copolymers'
  functional.
\newblock {\em Interfaces Free Bound.}, 11:447--474, 2009.

\bibitem{stillinger83}
F.~H. Stillinger.
\newblock Variational model for micelle structure.
\newblock {\em J. Chem. Phys.}, 78:4654--4661, 1983.

\bibitem{strukov}
B.~A. Strukov and A.~P. Levanyuk.
\newblock {\em Ferroelectric Phenomena in Crystals: Physical Foundations}.
\newblock Springer, New York, 1998.

\bibitem{theil06}
F.~Theil.
\newblock A proof of crystallization in two dimensions.
\newblock {\em Comm. Math. Phys.}, 262:209--236, 2006.

\bibitem{vedmedenko}
E.~Y. Vedmedenko.
\newblock {\em Competing Interactions and Pattern Formation in Nanoworld}.
\newblock Wiley, Weinheim, Germany, 2007.

\bibitem{wagner83}
H.-J. Wagner.
\newblock Crystallinity in two dimensions: a note on a paper of {C}. {R}adin:
  ``{T}he ground state for soft disks'' [{J}. {S}tatist. {P}hys. {\bf 26}
  (1981), 365--373)].
\newblock {\em J. Stat. Phys.}, 33:523--526, 1983.

\bibitem{wigner34}
E.~Wigner.
\newblock On the interaction of electrons in metals.
\newblock {\em Phys. Rev.}, 46:1002--1011, 1934.

\bibitem{yip06}
N.~K. Yip.
\newblock Structure of stable solutions of a one-dimensional variational
  problem.
\newblock {\em ESAIM Control Optim. Calc. Var.}, 12:721--751, 2006.

\end{thebibliography}

\bibliographystyle{plain}

\noindent {\sc Dorian Goldman\\
  Courant Institute of Mathematical
  Sciences, New York, NY 10012, USA,\\
  \& UPMC Univ  Paris 06, UMR 7598 Laboratoire Jacques-Louis Lions,\\
  Paris, F-75005 France ;\\
  CNRS, UMR 7598 LJLL, Paris, F-75005 France \\
  {\tt      dgoldman@cims.nyu.edu} \\
  \\
  {\sc Cyrill B. Muratov} \\
  Department of Mathematical Sciences, \\ New Jersey Institute of
  Technology, \\ Newark,
  NJ 07102, USA \\
  {\tt muratov@njit.edu}
  \\  \\
  {\sc Sylvia Serfaty}\\
  UPMC Univ  Paris 06, UMR 7598 Laboratoire Jacques-Louis Lions,\\
  Paris, F-75005 France ;\\
  CNRS, UMR 7598 LJLL, Paris, F-75005 France \\
  \&  Courant Institute, New York University\\
  251 Mercer st, NY NY 10012, USA\\
  {\tt serfaty@ann.jussieu.fr}

\end{document}